\documentclass[preprint]{vldb}
\pdfoutput=1
\usepackage{hyperref}
\usepackage[vldb,techreport]{manuprep}

\usepackage[aboveskip=0pt,belowskip=0pt]{subcaption}
\iftech
\usepackage[font=small, labelfont=bf, skip=0pt]{caption}
\else
\usepackage[font=scriptsize, labelfont=bf, skip=0pt]{caption}
\fi
\usepackage{float}
\makeatletter
\def\@copyrightspace{\relax}
\def\@copyrightblurb{\relax}
\def\@mkbibcitation{\relax}
\makeatother

\iftech
\else
\fi

\input{macros.tex}            %
\def\pma{\textsf{PMA}\xspace}
\def\pos{\textsf{PHOS}\xspace}
\def\ssi{\textsf{SSI}\xspace}

\def\rangetrim{{\sf RangeTrim}\xspace}
\def\optstop{{\sf OptStop}\xspace}

\def\dataset{\mathcal{D}\xspace}

\def\dslt#1{\dataset_{< #1}\xspace}
\def\dsgt#1{\dataset_{> #1}\xspace}

\def\ecdf{\widehat{F}\xspace}

\def\iter{i}

\def\whs{\widehat{\sigma}}

\def\lbound{\texttt{Lbound}\xspace}
\def\rbound{\texttt{Rbound}\xspace}

\def\hatg{\hat{g}}
\def\ghat{\hatg}
\def\gell{g_\ell}
\def\garr{g_r}
\def\interval{[\gell, \garr]}

\declarecommand{\numflightstuples}{606 million\xspace}

\declarecommand{\flights}{\textsc{Flights}\xspace}
\declarecommand{\sensor}{\textsc{Sensor}\xspace}
\declarecommand{\taxi}{\textsc{Taxi}\xspace}
\declarecommand{\police}{\textsc{Police}\xspace}
\declarecommand{\synth}{\textsc{Synthetic}\xspace}

\def\hoef{\textsf{Hoeffding}\xspace}
\def\hoefrt{\textsf{Hoeffding+RT}\xspace}
\def\bern{\textsf{Bernstein}\xspace}
\def\bernrt{\textsf{Bernstein+RT}\xspace}
\def\plusrt{\textsf{+RT}\xspace}
\def\exact{\textsf{Exact}\xspace}

\def\peek{\textsf{ActivePeek}\xspace}
\def\sync{\textsf{ActiveSync}\xspace}
\def\scan{\textsf{Scan}\xspace}

\def\topseparated#1{top-$#1$ separated\xspace}
\def\bottomseparated#1{bottom-$#1$ separated\xspace}
\def\ordered{groups ordered\xspace}

\def\intervalsmallerthan#1{$\hat{g}_r-\hat{g}_\ell < #1$}
\def\intervalgeq#1{$\hat{g}_r-\hat{g}_\ell \geq #1$}
\def\notinside#1{$#1\notin\interval$}
\def\isinside#1{$#1\in\interval$}
\def\enough{$c\geq m$}
\def\relintervalsmallerthan#1{$\max\{\frac{g_r-\hat{g}}{g_r}, \frac{\hat{g}-g_\ell)}{g_\ell}\} < #1$}
\def\relintervalgeq#1{$\max\{\frac{g_r-\hat{g}}{g_r}, \frac{\hat{g}-g_\ell)}{g_\ell}\} \geq #1$}

\newcounter{fqnum}
\newcommand{\qref}[1]{\hyperref[#1]{$q_{\ref*{#1}}$}}

\makeatletter
\newcommand{\fqtechlabel}[1]{\refstepcounter{fqnum}\label{external:#1}\hypertarget{#1}{}}
\newcommand{\fqlabel}[1]{\refstepcounter{fqnum}\label{internal:#1}\hypertarget{#1}{}}

\newcommand{\fqref}[1]{%
\@ifundefined{r@internal:#1}{%
\href{https://smacke.net/papers/ddavg.pdf}{F-q\ref*{external:#1}}
}{%
\protect\hyperlink{#1}{F-q\ref*{internal:#1}}%
}%
}

\newcommand{\fqrefp}[2]{
\@ifundefined{r@internal:#1}{%
\href{https://smacke.net/papers/ddavg.pdf}{\protect\hyperlink{#1}{F-q\ref*{external:#1}\blue{[#2\normalsize]}}}%
}{%
\protect\hyperlink{#1}{F-q\ref*{internal:#1}}\blue{[#2\normalsize]}%
}%
}

\makeatother

\newcommand{\fqrefstar}[1]{F-q\ref*{#1}}

\newenvironment{flightsquery}[2][]{%
\fqlabel{#1}\def\qheader{\scriptsize\color{pinegreen}{\# \fqrefstar{internal:#1}: #2}}%
}{}

\def\maxexactspeedup{more than $1000\times$\xspace}
\def\maxhoefspeedup{up to $124\times$\xspace}

\def\mindeptime{\blue{\texttt{\$min\_dep\_time}}\xspace}

\definecolor{codegreen}{rgb}{0,0.6,0}
\definecolor{codeblue}{rgb}{0,0,1.0}
\definecolor{codegray}{rgb}{0.5,0.5,0.5}
\definecolor{codepurple}{HTML}{C42043}
\definecolor{backcolour}{HTML}{F2F2F2}
\definecolor{bookColor}{cmyk}{0,0,0,0.90}

\newcommand\Small{\fontsize{8.5}{8.5}\selectfont}

\newcommand{\sql}[1]{\textsf{#1}}

\newcommand*\sqlidfont{\Small\ttfamily}
\newcommand*\sqlkwfont{\Small\ttfamily}
\newcommand*\sqlnumfont{\Small}
\newcommand*\sqlstrfont{\Small}

\lstdefinestyle{sqlstyle}{
    language=SQL,
    commentstyle=\color{pinegreen},
    keywordstyle=\color{codepurple}\sqlkwfont,
    identifierstyle=\sqlidfont,
    stringstyle=\color{codegreen}\sqlstrfont,
    numberstyle=\color{codegreen}\sqlnumfont,
    breakatwhitespace=false,
    breaklines=true,
    captionpos=b,
    keepspaces=true,
    numbersep=7pt,
    showspaces=false,
    showstringspaces=false,
    showtabs=false,
    framexrightmargin=-4pt,
    xleftmargin=4pt,
    framerule=0.5pt,
    escapeinside={<@}{@>},
    basewidth = {\thesissub{.5em}{.42em}},
    columns=flexible,
    basicstyle=\Small
}

\lstnewenvironment{lstsql}[1][]{
\lstset{style=sqlstyle}%
\iftech
\lstset{frame=tblr}
\fi
}{}

\declarecommand{\red}[1]{\textcolor{red}{#1}}
\declarecommand{\blue}[1]{\textcolor{blue}{#1}}
\declarecommand{\cyan}[1]{\textcolor{cyan}{#1}}
\declarecommand{\green}[1]{\textcolor{green}{#1}}
\declarecommand{\gray}[1]{\textcolor{gray}{#1}}
\declarecommand{\darkgreen}[1]{\textcolor{OliveGreen}{#1}}
\declarecommand{\pinegreen}[1]{\textcolor{pinegreen}{#1}}
\declarecommand{\purple}[1]{\textcolor{RoyalPurple}{#1}}

\ifnotes
\input{notes_macros.tex}
\else
\declarecommand{\agp}[1]{\ignorespaces}
\declarecommand{\agpquote}[2]{#1}
\declarecommand{\agpres}[1]{\ignorespaces}
\declarecommand{\smacke}[1]{\ignorespaces}
\declarecommand{\smackeres}[1]{\ignorespaces}
\declarecommand{\maryam}[1]{\ignorespaces}
\declarecommand{\ronitt}[1]{\ignorespaces}
\declarecommand{\agpst}[2]{#2}
\declarecommand{\smackout}[2]{\red{#2}}
\declarecommand{\smackoutres}[2]{\red{#2}}
\declarecommand{\resolved}[1]{\ignorespaces}
\declarecommand{\resolvedres}[1]{\ignorespaces}
\declarecommand{\strike}[1]{\ignorespaces}

\fi

\newcommand{\thesistechdel}[1]{#1}
\newcommand{\thesistechsub}[2]{#1}

\title{Rapid Approximate Aggregation \\ with Distribution-Sensitive Interval Guarantees\techreport{\LARGE\\\string[Technical Report\string]}}

\ifvldb
\author{
\alignauthor Stephen Macke$^{1,2}$ \quad\quad Maryam Aliakbarpour$^3$ \quad\quad Ilias Diakonikolas$^4$ \\ Aditya Parameswaran$^2$ \quad\quad Ronitt Rubinfeld$^3$\\
       \affaddr{\scalebox{1.0}{$^1$University of Illinois (UIUC) \quad\quad $^2$UC Berkeley \quad\quad $^3$MIT \quad\quad $^4$University of Wisconsin, Madison}} \\
       \affaddr{\scalebox{1.0}{\string{smacke,adityagp\string}@berkeley.edu 
       \quad \string{maryama@,ronitt@csail.\string}mit.edu 
       \quad ilias.diakonikolas@gmail.com}}
}
\else
\author{\IEEEauthorblockN{
Stephen Macke\IEEEauthorrefmark{1}\IEEEauthorrefmark{2},
Maryam Aliakbarpour\IEEEauthorrefmark{3},
Ilias Diakonikolas\IEEEauthorrefmark{4},
Aditya Parameswaran\IEEEauthorrefmark{2},
Ronitt Rubinfeld\IEEEauthorrefmark{3}}
\IEEEauthorblockA{
\IEEEauthorrefmark{1}University of Illinois (UIUC)
\quad\quad \IEEEauthorrefmark{2}UC Berkeley
\quad\quad \IEEEauthorrefmark{3}MIT
\quad\quad \IEEEauthorrefmark{4}University of Wisconsin, Madison}
\texttt{\string{smacke,adityagp\string}@berkeley.edu
\quad \string{maryama@,ronitt@csail.\string}mit.edu
\quad ilias.diakonikolas@gmail.com}}
\fi

\begin{document}

\maketitle

\def\root{.}
\def\figs{./}
\begin{abstract}
Aggregating data %
is fundamental 
to data analytics, data exploration, and OLAP. 
Approximate query processing (AQP) techniques
are often used to accelerate computation of aggregates
using samples, for which
confidence intervals (CIs) are widely used
to quantify the associated error.
CIs used in practice fall into two categories:
techniques that are {\em tight but not correct}, i.e., 
they yield tight intervals but only offer asymptotic guarantees, making them unreliable,
or techniques that are {\em correct but not tight}, i.e.,
they offer rigorous guarantees, but are overly conservative, leading to 
confidence intervals that are too loose to be useful. 
In this \work,
we develop a CI technique
that is both correct and tighter than traditional approaches. %
\agpquote{Starting from conservative CIs, we identify two issues
they often face: 
{\em pessimistic
mass allocation} (\pma) and 
{\em phantom outlier sensitivity} (\pos).
By developing a novel {\em range-trimming} 
technique for eliminating \pos
and pairing it with known CI techniques without \pma,
we develop a technique for computing CIs
with strong guarantees 
that requires fewer samples for the same width.}{You could techreport this to save space.}
We implement our techniques underneath a sampling-optimized in-memory column store
and show how to accelerate queries involving aggregates
on \techreport{real and synthetic datasets}
\papertext{a real dataset}
with speedups of \maxhoefspeedup over traditional AQP-with-guarantees
and \maxexactspeedup over exact methods.

\end{abstract}

\section{Introduction}
\label{sec:intro:dd}

Primitives for aggregation like
\AVG, \SUM, and \COUNT are key to making sense of 
and drawing insights from large volumes of
data, powering applications in 
OLAP, exploratory data analysis, 
and visual analytics.
Accelerating their computation 
is therefore of great importance.
Approximate Query Processing (AQP) is commonly used to accelerate
computation of these aggregates by estimating
them on a subset or sample of the full data. 
Reasoning about the error of the estimates
as introduced by approximation is crucial:
consumers of approximate answers---ranging 
from human decision makers to 
automated processes---rely on confidence
intervals (CIs) or error bounds as the foundation for
understanding the quality of the approximate answer.
Therefore, many AQP techniques
come with CIs to allow for 
more confident or informed decisions made 
using approximate estimates.

Error bounding,
or CI computation techniques 
take a confidence parameter $\delta\in[0,1]$, with the semantics that
the returned intervals $[g_\ell, g_r]$
fail to enclose the true aggregate $g^\star$ at most $\delta$
of the time.
One can tune $\delta$ to be as small as needed
($\delta=10^{-15}$ throughout this \work),
at the cost of requiring more samples
to achieve the same interval width $(g_r-g_\ell)$.
Likewise, for a given $\delta$,
taking more samples typically causes the error bounding procedure
to return a narrower confidence interval.
Since $\delta$ is typically small, we use the
phrase ``with high probability'' (\whp)
\agp{OK you do use it later. it's odd to have both with high probability and whp
introduced here at the same time. maybe use the former in the intro, introduce
whp later in the formal section -- with the hope that even systems folks will
follow all of the intro.}
as shorthand for ``with probability
greater than $(1-\delta)$''.
CI computation techniques need to satisfy two goals:
{\bf (i) compactness:} {\em by minimizing the
interval width $g_r - g_\ell$}, and {\bf (ii) correctness:} 
{\em by ensuring that
$g^\star \in [g_\ell, g_r]$ with high probability.}
However, achieving both compactness and correctness is difficult.

We outline the shortcomings of existing techniques,
that either prefer compactness over correctness ({\em asymptotic} techniques),
or correctness over compactness ({\em conservative} techniques),
below:

\begin{figure}[t]
\begin{lstsql}
SELECT Origin, AVG(DepDelay) FROM flights
GROUP BY Origin HAVING AVG(DepDelay) < 0
\end{lstsql}
\vspace{-5pt}
\caption{Origin airports with negative average delay. In this query,
the \AVG aggregates are consumed both by the user and by the system.}
\label{fig:putative:dd}
\end{figure}

\topic{Compactness without Correctness}
{\em Asymptotic} error bounding techniques such as bootstrap
CIs~\techsub{\cite{efron1979bootstrap,zeng2014analytical}}{\cite{efron1979bootstrap,efron1992bootstrap,zeng2014analytical}}
or central limit theorem (CLT)-based
CIs~\cite{student1908probable,hajek1960limiting} %
make assumptions about the distribution taken by the data given a
``large enough'' sample size. These procedures
typically give CIs that are much tighter (and therefore more
useful for drawing inferences about the query results), and have
enjoyed numerous applications in database and visual analytics
systems~\cite{pol2005relational,mozafari2017approximate,mozafari2015handbook,li2018approximate,fisher2011incremental,kwon2017sampling},
including Aqua~\cite{acharya1999aqua}, BlinkDB~\cite{blinkdb,agarwal2012blink},
DBO~\cite{jermaine2008scalable}, and online aggregation~\cite{Hellerstein1997}, and
have furthermore seen a number of DBMS-specific
extensions~\cite{zeng2014analytical,park2018verdictdb}.

However, these
asymptotic techniques result in intervals that only enclose the true
aggregate \whp in the limit as the size of the sample grows to
infinity.\footnote{The error of CLT-based
methods shrinks as $\bigo{1/\sqrt{m}}$,
but with constants depending on unknowns such as the third absolute normalized moment,
according to the Berry-Esseen theorem~\cite{berry1941accuracy,esseen1956moment}.}
That is, these techniques are correct
in the limit as the sample size approaches infinity,
but they provide no real guarantees for any given finite instance,
potentially leading to failures downstream.
For example, consider the query in \Cref{fig:putative:dd}, which determines
origin airports whose departing flights are ahead-of-schedule, on average. An
AQP system could use CIs to facilitate early stopping by using them to infer
on which side of the \sql{HAVING} threshold the various groups appear.
If such a system relies on asymptotic CIs, it is prone to serious
types of error, called {\em subset error} and
{\em superset error}~\cite{mozafari2017approximate}, whereby certain tuples
may be missing, and other tuples may appear spuriously.

\topic{Correctness without Compactness}
Recognizing the downsides of asymptotic approaches,
recent work~\cite{SPS,alabi2016pfunk,Kim2015,rahman2016ve,macke2018adaptive}
has begun to adopt {\em conservative} error bounders, which leverage
concentration inequalities to compute CIs.
These procedures return bounds that follow {\em probably approximately correct}
(PAC)~\cite{valiant2013probably} semantics:
given $\delta \in [0,1]$, the probability that the procedure returns lower and
upper bounds $[g_\ell, g_r]$ around the approximate aggregate
$\hat{g}$ that fail to enclose the true aggregate $g^\star$
should be {\em at most} $\delta$ for {\em any} sample size
(in contrast with asymptotic techniques,
for which the probability converges to $\delta$ given a large enough sample).
These techniques have been used in online aggregation~\cite{Hellerstein1997,haas1996hoeffding}
and more recently in work on visual analytics~\cite{alabi2016pfunk,Kim2015,rahman2016ve,macke2018adaptive}.

\begin{figure}[t]
\begin{densecenter}
\includegraphics[width=\thesissub{1.01}{.85}\linewidth]{\figs/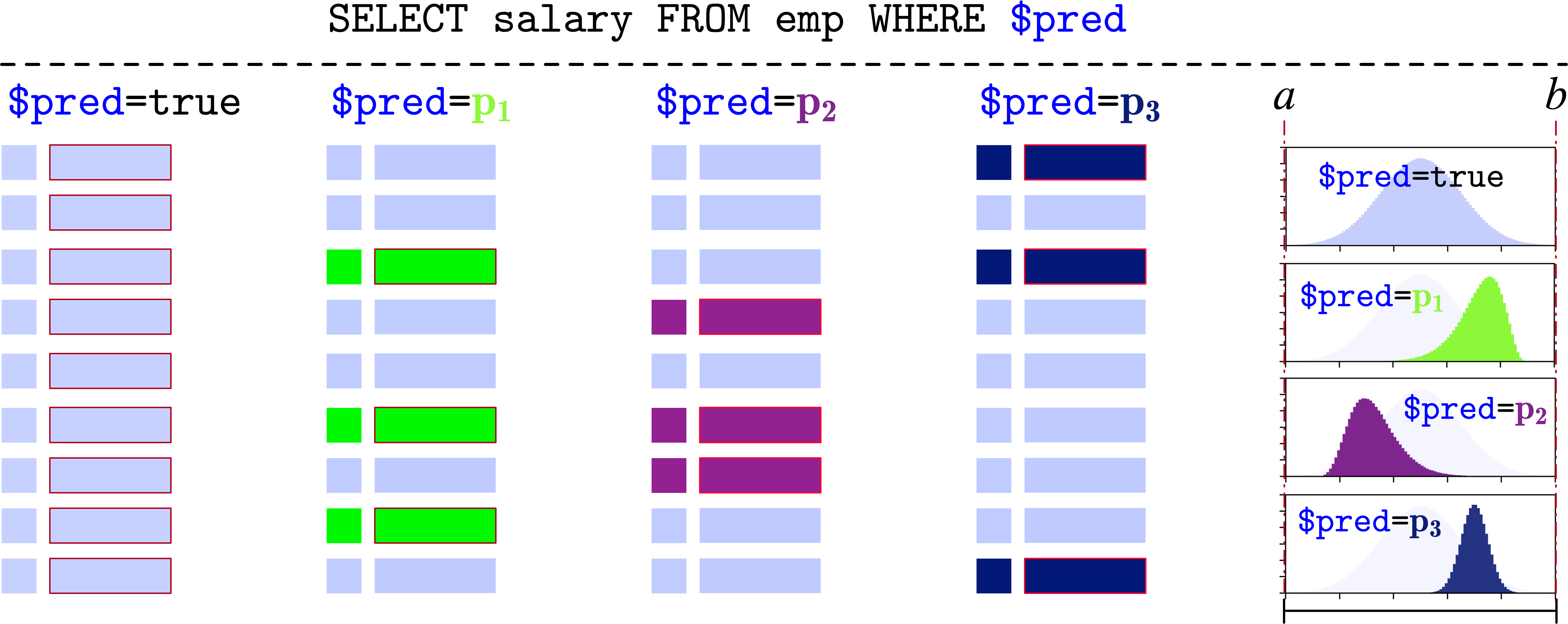}
\caption{Few points may lie near the range bounds $a$ and $b$, and with filters applied,
the true range could be significantly smaller than $(b-a)$.}
\label{fig:predrange:dd}
\end{densecenter}
\end{figure}

In general, conservative methods such as those based on Hoeffding's
inequality~\cite{hoeffding1963probability} or on the
Hoeffding-Serfling inequality~\cite{serfling1974probability}
rely on a-priori knowledge of {\em range bounds}
$a$ and $b$ between which the data fall (typically inferred during data loading). 
Although they achieve the correctness goal of error bounders,
when used for \AVG,
the CI width for Hoeffding-based error bounders
scales with the range size $(b-a)$, creating at least two major issues in the
context of a relational database, illustrated in \Cref{fig:predrange:dd}.
{\bf \em (i) First, the presence of a very few outliers can significantly widen the range $[a, b]$}
(and therefore the CI width), even though most of the data may lie in a much smaller
effective range. 
\techreport{In \Cref{fig:predrange:dd}, we see that even though the range of \texttt{salaries}
is $b-a$ when \texttt{\$pred = true}, most of the data is concentrated
in the center of the range. }
{\bf \em (ii) Second, predicates and groupings
may be applied during data exploration,
so that the filtered data lies in a smaller range than $[a, b]$}\techreport{;
in \Cref{fig:predrange:dd}, with \texttt{\$pred = $\mathbf{p_i}$}, 
we see that the range of filtered salaries is much smaller than 
even the \texttt{\$pred = true} case. 
However, direct application of Hoeffding-based
methods do not account for the tighter range of the filtered data,
instead treating the sampled tuples as if they were
taken from the original (unfiltered) data}.

\topic{Key Research Challenges and Contributions}
With this background in mind, this \work aims to
{\bf \em preserve correctness} (or safety) of conservative error bounders for
\AVG, \SUM, and \COUNT aggregates
{\bf \em while also providing compactness} (for speed).
We encounter a number of challenges toward this end:

\emtitle{1. Identifying conservative error bounder pathologies.}
To improve the viability of approaches with strict
correctness guarantees,
we must first determine the circumstances under which conservative
error bounders are {\em too} conservative, and understand when
fundamental limits prevent improvements without sacrificing guarantees.

\emtitle{Our contribution:}
We identify two issues in range-based
concentration inequalities
that cause unnecessary looseness
when used to compute conservative error bounds for \AVG. %
The first, {\em pessimistic mass allocation} (\pma),
refers to the unnecessary placement of unseen
probability mass at endpoints $a$ and $b$ of the range enclosing the data.
The second, {\em phantom outlier sensitivity} (\pos),
occurs when computation of the lower confidence bound $g_\ell$
depends on the upper range bound $b$ {\em even without
observed samples near $b$},
and vice versa for a dependency from $a$ to $g_r$.
\pos captures the intuition that unobserved
large (small) values should not loosen $g_\ell$ ($g_r$).

\emtitle{2. Correcting error bounder pathologies.}
After identifying correctable issues with existing conservative error
bounders, we need to develop novel techniques that address these
issues, while keeping in mind that these techniques should be
efficient in terms of computation and memory.

\emtitle{Our contribution:}
We develop a simple and general error bounding technique, {\em range trimming},
that corrects \pos without sacrificing desirable PAC semantics.
At a high level, range trimming operates by making
error bounders {\em asymmetric},
so that
$g_\ell$ depends only on the \MAX
value seen (and not on $b$), and
$g_r$ depends only on the
\MIN value seen, yielding tighter intervals
when $(\MAX-\MIN)$ is smaller than $(b-a)$.
Range trimming can be used with any existing conservative range-based
error bounder (\ie, an error bounder whose only assumption is that data falls in $[a, b]$).
We show how range trimming %
can be used to develop an error bounder for \AVG{}
(and by extension \SUM) with neither \pos nor \pma
by using it alongside a bounder based on Bernstein's inequality.

\emtitle{3. Minimizing sampling overhead.}
In order to enjoy the benefits of early termination for
queries with multiple aggregates, we need to ensure that
termination is not bottlenecked on any single aggregate,
allowing query processing to adaptively sample from the
most informative locations on physical storage while
simultaneously minimizing overhead.

\emtitle{Our contribution:}
We show how to couple our approach with a sampling-optimized column store
that takes without-replacement samples in a
locality-aware manner,
and that leverages bitmap indexes to prioritize samples that enable
earlier termination in the case of \sql{GROUP BY} clauses.
\techreport{Furthermore, although existing conservative
error bounders assume knowledge of the dataset size
(an unreasonable assumption when a filter of unknown selectivity is applied),
we show how to circumvent this limitation by
computing an upper bound on this size online.}

\topic{Impact}
We develop error bounding techniques
that more effectively leverage distributional
information of the underlying data, and that therefore
often lead to tighter error bounds as compared with those yielded by
typical conservative error bounders. %
When used in conjunction with a sampling-optimized column store for in-memory analytics,
we demonstrate speedups of \maxexactspeedup over exact techniques and
\maxhoefspeedup over traditional conservative approximate techniques,
all without sacrificing strong correctness guarantees.

\topic{Extensibility}
While our presentation focuses on confidence intervals
for queries over a single table with simple \AVG aggregates,
we note that our techniques are more general and can be used to
facilitate \SUM and \COUNT aggregates, queries over views formed from joins in
a snowflake schema, and queries with general UDFs --- we discuss these extensions
in \Cref{sec:scansample:dd}\thesistechsub{ and more in
\papertext{our extended technical report~\cite{techreport-dd}.}
\techreport{the appendix.}}{.}

\topic{Outline} The rest of this \work is organized as follows.
\Cref{sec:problem:dd} discusses existing conservative error bounders
and their prior usage in the DBMS literature, and
develops a conceptual framework for identifying issues with these
error bounders. In \Cref{sec:algorithm:dd} we develop the theory behind our \rangetrim
technique, and show how to fix issues with previous error bounders
in \Cref{sec:problem:dd}. \Cref{sec:system:dd} addresses systems
issues that appear when sampling without replacement
and develops \fv, our sampling-optimized column store,
and \Cref{sec:experiments:dd} empirically
evaluates our techniques in the context of this system.
We survey additional related work in \Cref{sec:related:dd}.

\section{DBMS Error Bound Integration}
\label{sec:problem:dd}
\begin{table}[t]
  \begin{densecenter}
  {\papertext{\scriptsize}\begin{tabular}{|C{\thesissub{2.3cm}{4.2cm}}|M{\thesissub{5.3cm}{11.2cm}}|}
    \hline
    \rowcolor[HTML]{C0C0C0} 
    {\bf Symbols / Terms} & {\bf Descriptions} \\ \hline
    $\dataset, N, S, \techreport{c,} m$ & 
    Dataset, \numof points in dataset (\ie $|\dataset|$), sample, \numof points \techreport{taken
    (for $c$) or desired (for $m$)} in sample (\ie $|S|$)\\ \hline
    $g^\star, \hat{g}, g_\ell, g_r$ & True aggregate, estimate, error bounds\\ \hline
    $a, b, \sigma^2, \widehat{\sigma}^2$
    $\delta, \veps$ & Range bounds, variance, empirical variance,
    error probability upper bound, error \\ \hline
    \techreport{$F, \ecdf, L, U$ & True / empirical CDF, lower and upper bounds on true CDF \\ \hline}
    \lbound, \rbound & Confidence lower (\resp upper) bounding routines parameterized on
                        $a, b, N$, and other sample state\techreport{(see \S\ref{sec:state:dd})}. \\ \hline
    \ssi, \pma, \pos & Sample-size-independent, 
                       pessimistic mass allocation,
                      phantom outlier sensitivity \\ \hline
    \end{tabular}}
  \caption{Glossary of terms / notation.}
  \label{tab:notation:dd}
  \end{densecenter}
\end{table}

In this section, we first
describe applications of confidence intervals
for facilitating query processing in a database system (\S\ref{sec:apps:dd}).
Next, we survey methods for computing error bounds with guarantees applicable
to DBMS aggregates (\S\ref{sec:computing-cis:dd})
identify their shortcomings (\S\ref{sec:pathologies:dd}) and conclude
with a formal problem statement (\S\ref{sec:prob-statement:dd}).

\subsection{DBMS CI Applications}
\label{sec:apps:dd}
Consider the query in \Cref{fig:putative:dd}.
In this query, \AVG aggregates are both {\em displayed} as output
in the query results, and are also used to {\em filter}
the set of tuples in the output. This reflects two major applications
of confidence intervals in a DBMS setting: CIs that are {\em explicitly used}
downstream, \ie, by an analyst, or CIs that are
{\em implicitly used} by automated processes.

\topic{Explicit Use of Downstream CIs}
When approximating aggregates in a DBMS,
confidence intervals can be included in the output displayed to users.
For this application, in which CIs are {\em explicitly} displayed to users,
the \AVG aggregates belonging to the groups output by the query in \Cref{fig:putative:dd}
are augmented with confidence intervals and included in the output.
This application helps users
{\em to reason about uncertainty in approximate answers}
and has seen prior usage in the database and visual
analytics literature~\cite{Hellerstein1997,fisher2011incremental}.

\topic{Implicit Use of Downstream CIs}
Confidence intervals have been applied toward facilitating
various other kinds of downstream applications, for example in order
to enable early stopping. Example applications from prior literature
include high-level accuracy
contracts~\cite{mozafari2015handbook,park2018verdictdb} %
(\ie, guaranteeing query results are
within $\veps$ of the correct), ranking query results~\cite{Kim2015,macke2018adaptive},
and bounding relative error~\cite{alabi2016pfunk}. In all cases, the user need
not ever observe the interval: the goal is to provide {\em early stopping}
while ensuring {\em correct results}. We consider these
applications later in our experiments in \Cref{sec:experiments:dd}.

\topic{Goal}
In this \work, we are primarily concerned with enabling CI {\em compactness} (to reduce query latency)
without sacrificing CI {\em correctness} (thereby ensuring safety),
for both explicit and implicit applications of CIs.
{\em The major goal is therefore to develop CI techniques that are as tight
as possible, while always enclosing the quantity in question.}
Throughout this section and \Cref{sec:algorithm:dd}, we will focus
our discussion on CIs for \AVG aggregates;
we will cover \SUM and \COUNT aggregates in \Cref{sec:system:dd}.

\subsection{Computing CIs in a DBMS}
\label{sec:computing-cis:dd}
We now describe methods for computing error bounds with accuracy guarantees
in a database system, along with any assumptions required.
Relevant notation is summarized in \Cref{tab:notation:dd}.
We begin by defining %
error bounders, bounds, and confidence intervals.

\begin{definition}[{\protect{$(1-\delta)$} error bounders and bounds}]
\label{def:bounds:dd}
A procedure $P$ that returns error bounds $\interval$ for some
aggregate $g^\star$ given a sample is a {\em $(1-\delta)$ error bounder}
if, across all possible samples,
$\Parg{g^\star \notin [g_\ell, g_r]} < \delta$.
$\interval$ is called the {\em $(1-\delta)$ confidence interval}
for $g^\star$, and
$g_\ell$ and $g_r$ are collectively referred to as
{\em $(1-\delta)$ error} or {\em confidence bounds}.
\end{definition}

In contrast with asymptotic error bounders that only satisfy
$\Parg{g^\star\notin[g_\ell, g_r]} \approx \delta$
for large-enough sample sizes,
the $(1-\delta)$ error bounders from \Cref{def:bounds:dd}
{\em always} satisfy $\Parg{g^\star\notin[g_\ell, g_r]} < \delta$ {\em for any sample size},
so we call them {\em sample-size-independent} (\ssi).

\subsubsection{Assumptions Applicable to Data in a DBMS}
\label{sec:assumptions:dd}

In the case of \AVG aggregates, all error bounding procedures
require some prior knowledge about the data
over which they operate -- otherwise, outliers can have
arbitrarily strong effects on the aggregate in question.
Weaker assumptions are more general, but
typically yield more conservative bounds.

In this \work, we make two assumptions about the data $\dataset$
over which queries operate:
first, that
every datapoint $x \in \dataset$ lies in some interval $[a, b]$;
second, that
datapoints can be effectively sampled {\em without replacement} from $\dataset$.
We now discuss these assumptions in the context of prior work and show
that they can be implemented effectively within real systems.

\topic{Known Range Bounds}
As in prior work~\cite{Hellerstein1997}, we assume that the database
catalog maintains {\em range bounds} $a$ and $b$ for the \MIN and
\MAX of each continuous
column, inferred, for example,
during data loading. (Note that we do not require $[a, b] = [\MIN, \MAX]$,
but only that $[a, b] \supseteq [\MIN, \MAX]$.)
These assumptions are more
applicable in the context of a database
as compared with stronger distributional assumptions (\eg, that the
data are normal or that they obey a tighter sub-Gaussian parameter than that implied
by the range bounds~\cite{wainwright2015basic})
and can be easily maintained in the case of insertions.
We refer to bounders that assume knowledge of $a$ and $b$
as {\em range-based error bounders} throughout this \work.
Furthermore, \thesistechdel{we show in \papertext{our extended technical report~\cite{techreport-dd}}
\techreport{the appendix (\S\ref{sec:exprrange:dd})}
that} it is possible to leverage the range assumption even in the case of
aggregates involving arbitrary expressions over multiple columns
by first solving an optimization problem
for derived range bounds $a'$ and $b'$ that enclose the transformed data.

\topic{Sampling Without Replacement}
\agp{Flow from previous paragraph is a bit abrupt.}
Estimates for \AVG aggregates generally converge faster for samples taken without replacement
than samples taken with replacement~\cite{serfling1974probability,bardenet2015concentration}.
In the context of a DBMS, sampling with replacement has traditionally
been considered easier than sampling without replacement, since the
system does not need to ``remember'' the samples already taken~\cite{olken1993random,Kim2015}.
Sampling as traditionally implemented, however, also has poor locality properties,
as nearly every read operation results in a cache miss. Another approach taken in prior
work~\cite{qin2014pf,wu2010continuous,zeng2015g,macke2018adaptive}
is to materialize samples ahead-of-time by performing a single up-front shuffle
of the entire relation, so that sampling without replacement
can be implemented via a scan of the data {\em regardless} of any applied filters or
other transformations. Since this approach is valid for multiple queries executed
during ad-hoc, exploratory workloads (in contrast with approaches that use workload
assumptions to pre-materialize stratified samples~\cite{icicles,blinkdb}),
we design our system architecture around this approach, described in more detail
in \Cref{sec:system:dd}.

\iftech
\subsubsection{State for DBMS Error Bounds}
\label{sec:state:dd}
OLAP queries must operate over many tuples, so it is desirable
that aggregations and their error bounders maintain small of memory footprints as possible
as new tuples are examined, although we will see in \Cref{sec:bounders:dd}
that some bounders must maintain state which grows with the number of
tuples examined. To better understand implementation details for error
bounders within the context of a DBMS,
we present error bounders in terms
of the following interface:
\begin{denseding}
{\color{linkpurple}\item\leavevmode\vphantom{a}\label{item:init}}\texttt{init\_state()}: Initializes state needed for error bounds.
{\color{linkpurple}\item\leavevmode\vphantom{b}\label{item:update}}\texttt{update\_state($S,v$)}: Given the current state $S$ and a newly-seen value $v$,
      compute new state $S'$.
{\color{linkpurple}\item\leavevmode\vphantom{c}\label{item:lbound}}\texttt{Lbound($S,a,b,N,\delta$)}: Return a confidence lower bound
              for a sample whose relevant statistics are captured in state $S$,
              assuming the sample came from a finite dataset $\dataset$ of $N$ values in $[a,b]$.
              The probability that the sample leads to this function returning a value greater
              than $\AVG(\dataset)$ is $< \delta$.
{\color{linkpurple}\item\leavevmode\vphantom{d}\label{item:rbound}}\texttt{Rbound($S,a,b,N,\delta$)}: Symmetric to \texttt{Lbound} for the confidence upper bound.
      Can typically be implemented in terms of \texttt{Lbound} after
      a suitable transformation of $S$.
\end{denseding}
The state $S$ captures information such as the count of tuples examined
and the current running average, as well as anything else required by
\texttt{Lbound} and \texttt{Rbound}. The state initialization and %
update logic is analogous to state maintenance logic for aggregates functions as
implemented in existing commercial database systems\techsub{~\cite{psqlagg}}{~\cite{mssqlagg,oracleagg,psqlagg}}.

Note that both \lbound and \rbound depend on the range bounds $a$ and $b$,
as well as the data size $N$ (allowing for tighter bounds when sampling without
replacement).
\fi

\subsubsection{Error Bounds for Finite and Bounded Data}
\label{sec:bounders:dd}
\techreport{\smacke{Add statements for applicable inequalities to the appendix.}}
In this section, we review some techniques for computing
confidence intervals that leverage only the assumptions
discussed in the previous subsection: that samples are
taken without-replacement from data bounded in some
a priori-known range $[a, b]$. Our goal is not to be
exhaustive but representative, drawing attention
to previous applications in the DB literature\techreport{ (and lack thereof)}.
Further details about these bounders, such as implementation pseudocode and full restatements
of relevant theorems, are available in our extended technical report~\cite{techreport-dd}.

\techreport{
\begin{algorithm}[t]
\caption{Hoeffding-Serfling error bounder~\protect\cite{serfling1974probability}}\label{alg:hsbounder:dd}
\algrenewcomment[1]{\texttt{/*} #1 \texttt{*/} \\}
\SetKwInOut{Input}{Input}
\SetKwInOut{Output}{Output}
\SetKwRepeat{Do}{do}{while}
\SetKwBlock{Repeat}{Repeat}{}
\SetKwProg{myfunc}{function}{}{}
\SetKwFunction{init}{init\_state}
\SetKwFunction{update}{update\_state}
\SetKwFunction{lbound}{Lbound}
\SetKwFunction{rbound}{Rbound}
{\thesisdel{\scriptsize}
\nonl\ \\
\myfunc{\init{}\quad\ref{item:init}}{
\Return{\upshape{\texttt{\big\{$m$: $0$, $\ghat$: $0$\big\}}}}\;
}
\nonl\ \\
\myfunc{\update{$S, v$}\quad\ref{item:update}}{
$m' \gets S.m + 1$\;
$\ghat' \gets S.\ghat + (v - S.\ghat) / m'$\;
\Return{\upshape{\texttt{\big\{$m$: $m'$, $\ghat$: $\ghat'$\big\}}}}\;
}
\nonl\ \\
\myfunc{\lbound{$S, a, b, N, \delta$}~\upshape{\cite{serfling1974probability}\quad\ref{item:lbound}}}{
$\veps \gets (b-a)\cdot\sqrt{\frac{\log{(1/\delta)}}{2\cdot S.m} \cdot(1 - \frac{S.m-1}{N})}$\;
\Return{\upshape{$S.\ghat - \veps$}}\;
}
\nonl\ \\
\myfunc{\rbound{$S, a, b, N, \delta$}\quad\ref{item:rbound}}{
$S.\ghat \gets (a+b) - S.\ghat$\;
\Return{\upshape{$(a+b) - \lbound(S, a, b, N, \delta)$}}\;
}
}
\end{algorithm}
}
\topic{Hoeffding-Serfling-based Bounder}
An error bounder based on the Hoeffding-Serfling inequality~\cite{serfling1974probability}
computes CIs whose widths depend only on the range $(b-a)$ and the number of samples
$m$, and that have size $\bigo{(b-a)/\sqrt{m}}$\techreport{ (if we ignore the sampling fraction term)}.
While asymptotically optimal for worst-case data distributed with half of the
points at $a$ and the other half at $b$, it is needlessly wide in practice, when
few points occur near $a$ or $b$.
\techreport{An implementation of this bounder in terms of our interface from \Cref{sec:state:dd}
is given in \Cref{alg:hsbounder:dd}.}
\techreport{We give a statement of the Hoeffding-Serfling inequality and derive
the corresponding error bounder.

\begin{lemma}[{Hoeffding-Serfling Inequality~\protect\cite{serfling1974probability}}]
\label{lem:hs:dd}
Let $\dataset = x_1, \ldots, x_N$ be a set of $N$ values in $[a, b]$
with average value $\AVG(\dataset) = \mu$. Let $X_1, \ldots, X_N$ be
a sequence of random variables drawn from $\dataset$ without replacement.
For every $1 \leq m \leq N$ and $\veps > 0$,
\[ \Parg{\max_{1 \leq k \leq m} \frac{\sum_{t=1}^k (X_t - \mu)}{N-k} \geq \frac{m\veps}{N-m}} \leq \delta \] where
\[ \delta = \exp{-\frac{2m\veps^2}{(1-\frac{m-1}{N})(b-a)^2}} \]
\end{lemma}
By focusing on $k=m$ and inverting the probability expression, we may compute
a $1-\delta$ lower confidence bound as
\[ \frac{1}{m}\sum_{t=1}^m X_t -
(b-a)\sqrt{\frac{(1-\frac{m-1}{N})(\log\frac{1}{\delta})}{2m}} \]
and likewise for a upper confidence bound (replacing ``$-$'' with ``$+$''), so that
$(1-\frac{\delta}{2})$ lower and upper confidence bounds may be
combined to yield a $(1-\delta)$ confidence interval (via a union bound).}

\techreport{
\begin{algorithm}[t]
\caption{\techreport{\fontsize{8.799}{8.799}\selectfont} Empirical Bernstein-Serfling err.\ bounder~\protect\cite{bardenet2015concentration}}\label{alg:ebsbounder:dd}
\algrenewcomment[1]{\texttt{/*} #1 \texttt{*/} \\}
\SetKwInOut{Input}{Input}
\SetKwInOut{Output}{Output}
\SetKwRepeat{Do}{do}{while}
\SetKwBlock{Repeat}{Repeat}{}
\SetKwProg{myfunc}{function}{}{}
\SetKwFunction{init}{init\_state}
\SetKwFunction{update}{update\_state}
\SetKwFunction{lbound}{Lbound}
\SetKwFunction{rbound}{Rbound}
{\thesisdel{\scriptsize}
\nonl\ \\
\myfunc{\init{}\quad\ref{item:init}}{
\Return{\upshape{\texttt{\big\{$m$: $0$, $\ghat$: $0$, $M_2$: $0$\big\}}}}\;
}
\nonl\ \\
\myfunc{\update{$S, v$}\quad\ref{item:update}}{
$m' \gets S.m + 1$\;
$\ghat' \gets S.\ghat + (v - S.\ghat) / m'$\;
$M_2' \gets S.M_2 + v^2$\;
\Return{\upshape{\texttt{\big\{$m$: $m'$, $\ghat$: $\ghat'$, $M_2$: $M_2'$\big\}}}}\;
}
\nonl\ \\
\myfunc{\lbound{$S, a, b, N, \delta$}~\upshape{\cite{bardenet2015concentration}\quad\ref{item:lbound}}}{
$\kappa \gets 7/3 + 3/\sqrt{2}$\;
$\rho \gets \indic{S.m\leq N/2}\cdot(1-\frac{S.m-1}{N})$\;
$\rho \gets \rho + \indic{S.m>N/2}\cdot((1-\frac{S.m}{N})\cdot(1+\frac{1}{S.m}))$\;
$\veps \gets \sqrt{S.M_2 / S.m - S.\ghat^2} \cdot \sqrt{\frac{2\rho \cdot \log{(5/\delta)}}{S.m}} + \kappa \cdot (b-a) \cdot \frac{\log{(5/\delta)}}{S.m}$\;
\Return{\upshape{$S.\ghat - \veps$}}\;
}
\nonl\ \\
\myfunc{\rbound{$S, a, b, N, \delta$}\quad\ref{item:rbound}}{
$S.\ghat \gets (a+b) - S.\ghat$\;
\Return{\upshape{$(a+b) - \lbound(S, a, b, N, \delta)$}}\;
}
}
\end{algorithm}
}
\topic{Empirical Bernstein-Serfling-based Bounder}
\iftech
A concentration inequality
for sampling without replacement
given in~\cite{bardenet2015concentration},
the {\em Bernstein-Serfling} inequality
assumes knowledge of both $(b-a)$ and
$\VAR(\dataset) = \sigma^2 = \frac{1}{N}\sum_{x\in\dataset}(x-\AVG(\dataset))^2$.
We defer a statement of the full result\thesisdel{ to the appendix}.
Here we note that inverting
the inequality gives error bounds as
\[ \frac{1}{m}\sum_{t=1}^m X_t \pm \bigo{\sigma/\sqrt{m} + (b-a)/m} \]
if we again ignore the sampling fraction term.
Comparing these error bounds
to those of Hoeffding-Serfling, which has widths of size $\bigo{(b-a)/\sqrt{m}}$
(again ignoring the sampling fraction),
we see that error bounds derived from the Bernstein-Serfling inequality can be
significantly tighter when $\sigma$ is small compared to $(b-a)$.

Knowledge of $\VAR(\dataset)$ typically cannot be
assumed in a setting where $\AVG(\dataset)$ is unknown.
Fortunately, there also exists an
{\em empirical} variant of the Bernstein-Serfling inequality (also given
in~\cite{bardenet2015concentration}, like the non-empirical variant).
The analysis for the empirical Bernstein-Serfling inequality proceeds by
augmenting the analysis for the non-empirical variant
with a concentration inequality relating the estimator
$\widehat{\sigma}^2 = \frac{1}{m}\sum_{t=1}^m(X_t - \bar{X})^2$ to $\VAR(\dataset)$.
We again deferring the full statement\thesisdel{ to the appendix}.
This yields $(1-\delta)$ \techreport{error bounds given by
\[ \frac{1}{m}\sum_{t=1}^m X_t \pm \bigo{\widehat{\sigma}/\sqrt{m} + (b-a)/m} \]}
\papertext{CIs whose widths have size $\bigo{\widehat{\sigma}/\sqrt{m} + (b-a)/m}$.}
Note that these error bounds differ from the those of the
non-empirical variant only in that $\sigma$ is
replaced by $\widehat{\sigma}$ (modulo slightly worse constants hidden
by the asymptotic notation).
Although $\widehat{\sigma}$ is a random quantity,
it concentrates near $\sigma$, so that an error bounder
based on the empirical Bernstein-Serfling bound returns bounds of asymptotically the same
width as those returned by an error bounder based on the non-empirical
variant and with full access to $\sigma^2$, \whp. \Cref{alg:ebsbounder:dd}
gives an implementation of an empirical
Bernstein-Serfling-based error bounder
in terms of our interface from \Cref{sec:state:dd}.
Note that \Cref{alg:ebsbounder:dd} as presented
shows computation of the sample variance in terms
of the second moment $M_2=\sum v^2$ for the sake of exposition;
a real implementation might use a more numerically stable one-pass
algorithm for the
variance~\cite{welford1962note,chan1983algorithms,ling1974comparison}.
\else
The {\em empirical Bernstein-Serfling inequality}
is another concentration inequality for sampling
without replacement, given in~\cite{bardenet2015concentration}.
A corresponding error bounder yields $(1-\delta)$
CIs whose widths have size $\bigo{\widehat{\sigma}/\sqrt{m} + (b-a)/m}$,
where $\widehat{\sigma}^2 = \frac{1}{m}\sum_{t=1}^m(X_t - \bar{X})^2$
is the maximum likelihood estimator for $\VAR(\dataset)$.
Although $\widehat{\sigma}$ is a random quantity, it concentrates
near $\sigma$; thus,
we see that error bounds derived from the Bernstein-Serfling inequality can be
significantly tighter than those derived from Hoeffding-Serfling
(which has widths of size $\bigo{(b-a)/\sqrt{m}}$)
when $\sigma$ is small compared to $(b-a)$.
In fact, we will see shortly (\S\ref{sec:pathologies:dd})
that Bernstein-based bounders do not suffer from one of the
problems that causes Hoeffding-based bounders to be {\em overly-conservative}.
\fi

\techreport{
\begin{algorithm}[t]
\caption{\techreport{\small} Anderson/DKW error bounder~\protect\cite{dvoretzky1956asymptotic,anderson1969confidence,massart1990tight}}\label{alg:dkwbounder:dd}
\algrenewcomment[1]{\texttt{/*} #1 \texttt{*/} \\}
\SetKwInOut{Input}{Input}
\SetKwInOut{Output}{Output}
\SetKwRepeat{Do}{do}{while}
\SetKwBlock{Repeat}{Repeat}{}
\SetKwProg{myfunc}{function}{}{}
\SetKwFunction{init}{init\_state}
\SetKwFunction{update}{update\_state}
\SetKwFunction{lbound}{Lbound}
\SetKwFunction{rbound}{Rbound}
{\thesisdel{\scriptsize}
\nonl\ \\
\myfunc{\init{}\quad\ref{item:init}}{
\Return{$\{\}$}
}
\nonl\ \\
\myfunc{\update{$S, v$}\quad\ref{item:update}}{
\Return{$S\cup \{v\}$}
}
\nonl\ \\
\myfunc{\lbound{$S, a, b, N, \delta$}\quad\ref{item:lbound}}{
$\veps \gets \sqrt{\frac{\log{(1/\delta)}}{2\cdot |S|}}$\;
$\ecdf \gets \text{ empirical CDF based on $S$}$\;
$S' \gets \{x \in S : \ecdf(x) \leq 1-\veps\}$\;
\Return{\upshape{$\veps\cdot a + (1-\veps)\cdot\AVG(S')$}}\;
}
\nonl\ \\
\myfunc{\rbound{$S, a, b, N, \delta$}\quad\ref{item:rbound}}{
\Return{\upshape{$(a+b) - \lbound((a+b)-S, a, b, N, \delta)$}}\;
}
}
\end{algorithm}
}
\topic{Anderson/DKW-based Bounder}
Anderson described a way to compute distribution-free / nonparametric error bounds for the mean
given error bounds for the cumulative distribution function (CDF) in~\cite{anderson1969confidence}.
\thesistechdel{\papertext{We describe the high-level idea behind this error bounder
in our extended technical report~\cite{techreport-dd}.}}
\techreport{Denoting the true and empirical CDF
for some distribution
supported on $[a,b]$ with $F$ and $\ecdf$, respectively,
Anderson showed how to use high-probability bounds $\alpha$ and $\beta$ such that
\[ \ecdf - \alpha \preceq F \preceq \ecdf + \beta \]
to get high-probability bounds on the mean of $F$.
To see how, recall the following identity:
\begin{restatable}{lemma}{cdfmean}
\label{lem:cdfmean:dd}
Consider a CDF $F$ supported on $[a,b]$. Then the mean $\mu$
of the distribution corresponding to $F$ satisfies
\[ \mu = b - \int_a^b F(x)dx \]
\end{restatable}
\noindent Thus, given lower and upper bounds $L$ and $U$ on the CDF $F$ that satisfy
$\forall x\in[a,b], L(x) \leq F(x) \leq U(x)$, error bounds around the
mean may be computed as
\[ \left[b - \int_a^b U(x)dx\ ,\quad b-\int_a^b L(x)dx\right] \]
since $L \preceq F \preceq U$ implies %
$-U \preceq -F \preceq -L$.

Anderson used 
the Dvoretzky-Kiefer-Wolfowitz (DKW) inequality~\cite{dvoretzky1956asymptotic} to
compute $\alpha$ and $\beta$.
Informally, DKW states that the empirical CDF $\ecdf$ computed
from \iid samples taken from a distribution with CDF $F$ concentrates around $F$ everywhere:

\begin{lemma}[{DKW Inequality~\protect\cite{dvoretzky1956asymptotic,massart1990tight}}]
\label{lem:dkw:dd}
Let $X_1, \ldots, X_m \iidsim F$, and let $\ecdf$ be the empirical CDF
corresponding to the sample $\{X_i\}$. Then for every $\veps > 0$,
\[ \Parg{\sup_{t\in\text{dom}(F)}|\ecdf(t) - F(t)| > \veps} \leq 2\exp{-2m\veps^2} \]
\end{lemma}

The DKW inequality provides a method to
obtain the values of $\alpha$ and $\beta$, since it implies that
\[ \ecdf - \sqrt{\frac{\log{2/\delta}}{2m}} \preceq F \preceq \ecdf + \sqrt{\frac{\log{2/\delta}}{2m}}\]
with probability greater than $1-\delta$.
At the time~\cite{anderson1969confidence} was published, however,
the constant in front of the DKW inequality had not yet been proved by Massart~\cite{massart1990tight},
so it appears that Anderson computed $\alpha$ and $\beta$ using a lookup table.

Although \Cref{lem:dkw:dd} as stated applies for sampling with replacement from an infinite population,
please see \Cref{sec:dkwnoreplace:dd}
for a proof that DKW still holds when $X_1, \ldots, X_m$ are drawn without
replacement from a finite population of size $N$, for any $N > 0$, stated
as the following theorem:
\begin{restatable}{theorem}{dkwnoreplace}
\label{THM:DKWNOREPLACE:DD} %
For any $N > 0$, the DKW inequality applies for sampling without replacement
from a finite dataset of size $N$.
\end{restatable} 
\noindent The procedure just described
for computing error bounds around the mean of a distribution given \iid samples thus also
works for computing error bounds around $\AVG(\dataset)$ given without-replacement samples
from the finite dataset $\dataset$. It is presented in terms of our interface from
\Cref{sec:state:dd} in \Cref{alg:dkwbounder:dd}.}

\topic{Applications in Prior DB Literature}
To our knowledge, Hoeffding and Hoeffding-Serfling-based bounders are the only
\ssi bounders that have
seen extensive use in the DB literature for computing
error bounds for \AVG~\cite{Kim2015,alabi2016pfunk,Hellerstein1997,haas1996hoeffding}.
We are aware of one incorrect application of the empirical Bernstein-Serfling
inequality~\cite{chen2016effective} (incorrect because the procedure
given in~\cite{chen2016effective} continuously recomputes confidence
$(1-\delta)$ intervals as more samples are taken, so that the overall procedure is no longer guaranteed
to fail with probability at most $\delta$). Overall, we consider it somewhat
surprising that error bounders derived from
the empirical Bernstein-Serfling inequality from~\cite{bardenet2015concentration}
have not seen more widespread usage, as they are nearly as simple to compute as
those derived from the Hoeffding-Serfling inequality and typically yield error bounds that are
much tighter.

\subsection{Error Bounder Pathologies}
\label{sec:pathologies:dd}
\begin{table}[t]
\begin{densecenter}
\begin{tabular}{|c|c|c|c|c|}
\hline
\rowcolor[HTML]{C0C0C0} 
\textbf{Error Bounder}              & \pma                 & \pos                 & Sampling         & Memory     \\ \hline
Hoeffding(-Serfling)                & \checkmark           & \checkmark           & \textsf{R* (NR)} & $\bigo{1}$ \\ \hline
Berstein(-Serfling)                 & \cellcolor{green!25} & \checkmark           & \textsf{R* (NR)} & $\bigo{1}$ \\ \hline
Anderson/DKW                        & \checkmark           & \cellcolor{green!25} & \textsf{R, NR}   & $\bigo{m}$ \\ \hline
\end{tabular}
\caption{Summary of \techsub{error bounder properties}{properties exhibited by various error bounders}.
\textsf{R} = sampling with replacement, \textsf{NR} = without. A \textsf{*}
indicates that the non-Serfling variant also holds for \textsf{NR} sampling.}
\label{tab:pathologies:dd}
\end{densecenter}
\end{table}

As a major technical contribution of this work, we identify
two problems that cause \ssi error bounders to be
{\em too} conservative. These pathologies, which we refer to
as {\em pessimistic mass allocation (\pma)} and {\em phantom outlier sensitivity (\pos)},
are based on simple intuitions about how
error bounders should behave: namely, they should return tighter
bounds when observing samples with fewer extreme values, and error lower bounds
(respectively error upper bounds) should only be looser due to potential
large values (\resp small values) if such values are actually observed.

\subsubsection{Pessimistic Mass Allocation}

\noindent \pma, defined as follows, captures the intuition that error bounders
should be sensitive to the observed sample values:
\begin{definition}[\pma]
\label{def:pma:dd}
An error bounding procedure $P$ exhibits {\em pessimistic mass allocation (\pma)}
if there exists a dataset $\dataset$ bounded in $[a, b]$,
a value $a'$ with $a < a' < b$,
and a set $S \subseteq \dataset$ with values in $[a, a')$ such that,
for $S' = \{\max{(x, a')} : x \in S\}$,
$P$ returns a confidence interval of the same width for both $S$ and $S'$.
$P$ likewise exhibits \pma if there exists some $b'$ with $a < b' < b$
and an $S$ with values in $(b', b]$ such that,
for $S' = \{\min{(x, b')} : x \in S\}$,
$P$ returns a confidence interval of the same width for both $S$ and $S'$.
\end{definition}
That is, an error bounder $P$ has \pma if it is possible to
replace the smallest (largest) elements in a sample with
something larger (\resp smaller) without shrinking the width of
$P$'s returned confidence interval.
Intuitively, $P$ is
overly-pessimistic about how mass in the \techreport{underlying} distribution
from which it is sampling is allocated, despite contrary evidence
observed in the sample.

\subsubsection{Phantom Outlier Sensitivity}

\noindent \pos captures the intuition that unobserved
extreme values should not affect both lower and upper error bounds:
\begin{definition}[\pos]
\label{def:pos:dd}
An error bounding procedure $P$ exhibits {\em phantom outlier sensitivity (\pos)}
if, for data falling in $[a, b]$, $P$'s returned confidence lower bound $g_\ell$
depends on the value of $b$, and similarly if the $g_r$ returned by $P$ depends on $a$.
\end{definition}
To understand \pos intuitively, consider the case of computing a confidence lower
bound. Given a sample $S$, %
the worse $P$ ``believes'' $S$ could be shifted (on average) toward larger values as compared to $\dataset$, %
the smaller of a confidence lower bound it should return. In what ways
could $S$ be shifted toward higher values? %
One possibility is if small elements
are underrepresented in $S$. The other possibility, and the one we
are interested in, is if large elements are overrepresented in $S$. For this reason,
a confidence lower bound should only be affected by datapoints near the upper
range bound $b$ {\em if it actually observes them}, and the appearance of $b$
in the computation of a confidence lower bound is a potential source of unnecessary
conservativeness.

\subsubsection{Examples of \pma and \pos in Error Bounders}

\techreport{\begin{figure}[t]
\begin{densecenter}
\includegraphics[width=\thesissub{.45\textwidth}{.85\linewidth}]{\figs/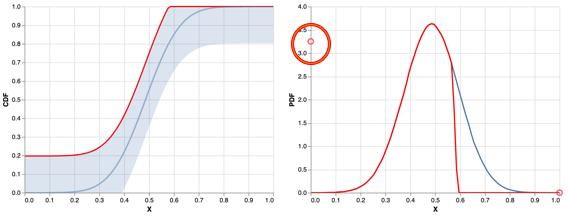}
\caption{Error bounds from the DKW inequality exhibit pessimistic mass allocation.
\smacke{Fix labels, bigger fonts, increased resolution}}
\label{fig:dkwpma:dd}
\end{densecenter}
\end{figure}}

In this section, we give examples of \pma and \pos in the context
of previously-discussed error bounders.
\Cref{tab:pathologies:dd} summarizes pathologies exhibited by
various \ssi error bounders.

\topic{Hoeffding-based} Hoeffding-based
error bounders suffer from both \pma and \pos. They have \pma
since their returned CIs have widths depending only on the range of the data, $(b-a)$,
and the number of samples. As such, replacing values in the sample
with larger or smaller values does not affect the width of the returned error bounds.
Such bounders also have \pos since
they have symmetric error, with both ends of the confidence interval
depending on both \techreport{range bounds} $a$ and $b$.

\topic{Berstein-based} Bernstein-based error bounders
do {\em not} suffer from \pma. To see this, notice that
increasing the smallest values in some sample will also reduce
the sample variance, affecting the width of the returned confidence interval,
and likewise for decreasing the largest values in the sample.
These bounders do, however, suffer from \pos. Like Hoeffding-based
bounders, they return confidence intervals with symmetric error, so that each end
of the confidence interval is affected by both ends of the data
range $a$ and $b$.

We will see in our experiments that these bounds can yield
significant speedups as compared with Hoeffding-based bounds,
when used to facilitate early termination of approximate queries.

\topic{Anderson/DKW-based} Anderson/DKW-based error bounders
are interesting in that they suffer from \pma, but not \pos.
\thesistechdel{\papertext{Please see our extended technical
report~\cite{techreport-dd} for a discussion.}}
\techreport{Consider the $\veps$ mass unaccounted for when computing a
confidence lower bound using an Anderson/DKW-based bounder. As shown
in \Cref{fig:dkwpma:dd}, it all goes toward to lower range bound, $a$,
which is sufficient for \pma. On the other hand, where does it come from?
It comes from the $\veps$-fraction largest observed points. This does not
depend at all on the value of the upper range bound $b$, indicating
that the confidence lower bound does not suffer from \pos. Symmetric statements
hold for the confidence upper bound, of course.}

\techreport{\smacke{Add others}}

\subsection{Problem Statement}
\label{sec:prob-statement:dd}
We are now ready to give a formal problem statement.
\begin{problem}
\label{prob:errorbounds:dd}
Design an \ssi error bounder that,
given a without-replacement sample from any $\dataset$
with elements from $[a,b] \subseteq \reals$,
suffers from neither \pma nor \pos when computing
$(1-\delta)$ error bounds for $\AVG(\dataset)$, for any $0<\delta<1$. %
\end{problem}
\vspace{-0.5em}
\techreport{\smacke{Mention unknown size since we might have filters of unknown selectivity?}}
Our solution to %
\Cref{prob:errorbounds:dd}
is given in \Cref{sec:algorithm:dd} and relies
on a technique we call {\em range trimming} in order to systematically eliminate \pos
from any range-based error bounder.

\techreport{Although the solution as presented in \Cref{sec:algorithm:dd} additionally
assumes knowledge of the size of $\dataset$, \Cref{sec:system:dd} shows how our real-world
implementation circumvents this limitation.}

\section{Fixing Bounder Pathologies}
\label{sec:algorithm:dd}
From our discussion in \Cref{sec:pathologies:dd}, we see that
there do exist error bounders with either \pma or \pos, but not both.
We first argue that error bounders without \pos must be {\em asymmetric};
that is, they cannot compute bounds of the form $\hat{g} \pm \veps$, where
the same $\veps$ is both added and subtracted to the sample average $\hat{g}$
in order to compute bounds. Next, we describe how to use a process we call
{\em range trimming} to convert any symmetric, ranged-based error bounder
to an asymmetric one without \pos.

\subsection{Decoupling Lower and Upper Bounds}

Excepting an error bounder based on DKW, all of the error bounders
surveyed suffer from \pos. This is because all the other error bounders
are based on concentration inequalities with {\em symmetric error} ---
that is, they return confidence intervals $[g_\ell, g_r]$ of the form
$[\hat{g}-\veps, \hat{g}+\veps]$. At a high level, it is precisely this
symmetry that causes \pos. Although a confidence lower bound should not
have any dependency on $b$, it is intuitively unavoidable that it has
some dependency on $a$. Reiterating, an estimate $\hat{g}$ could be
an overestimate because of (i) not enough observed values near $a$, or
(ii) too many observed values near $b$. A similar statement holds regarding
confidence upper bounds, with the roles of $a$ and $b$ reversed.

We hypothesize that it is impossible for any confidence lower bound (\resp upper bound)
to completely eliminate the dependency on $a$ (\resp $b$), since it is always possible
that the confidence bounding procedure got ``unlucky'' and operated on a sample in which
values near $a$ (\resp $b$) were underrepresented. Taking this hypothesis as given,
this means that any symmetric confidence bounding procedure that returns bounds of
the form $[\hat{g} - \veps, \hat{g} + \veps]$ will have $\veps$ dependent on both $a$ and $b$ ---
that is, any symmetric confidence bounding procedure will have \pos.
As such, the first step to eliminating \pos from range-based confidence bounders is
to accept asymmetric error as a hard requirement: that is, we must consider confidence bounding
procedures that return bounds of the form $[\hat{g} - \veps_\ell, \hat{g} + \veps_r]$ for
which $\veps_\ell$ and $\veps_r$ are not necessarily equal.

\begin{figure}[t]
\begin{densecenter}
\includegraphics[width=\thesissub{.45\textwidth}{.75\linewidth}]{\figs/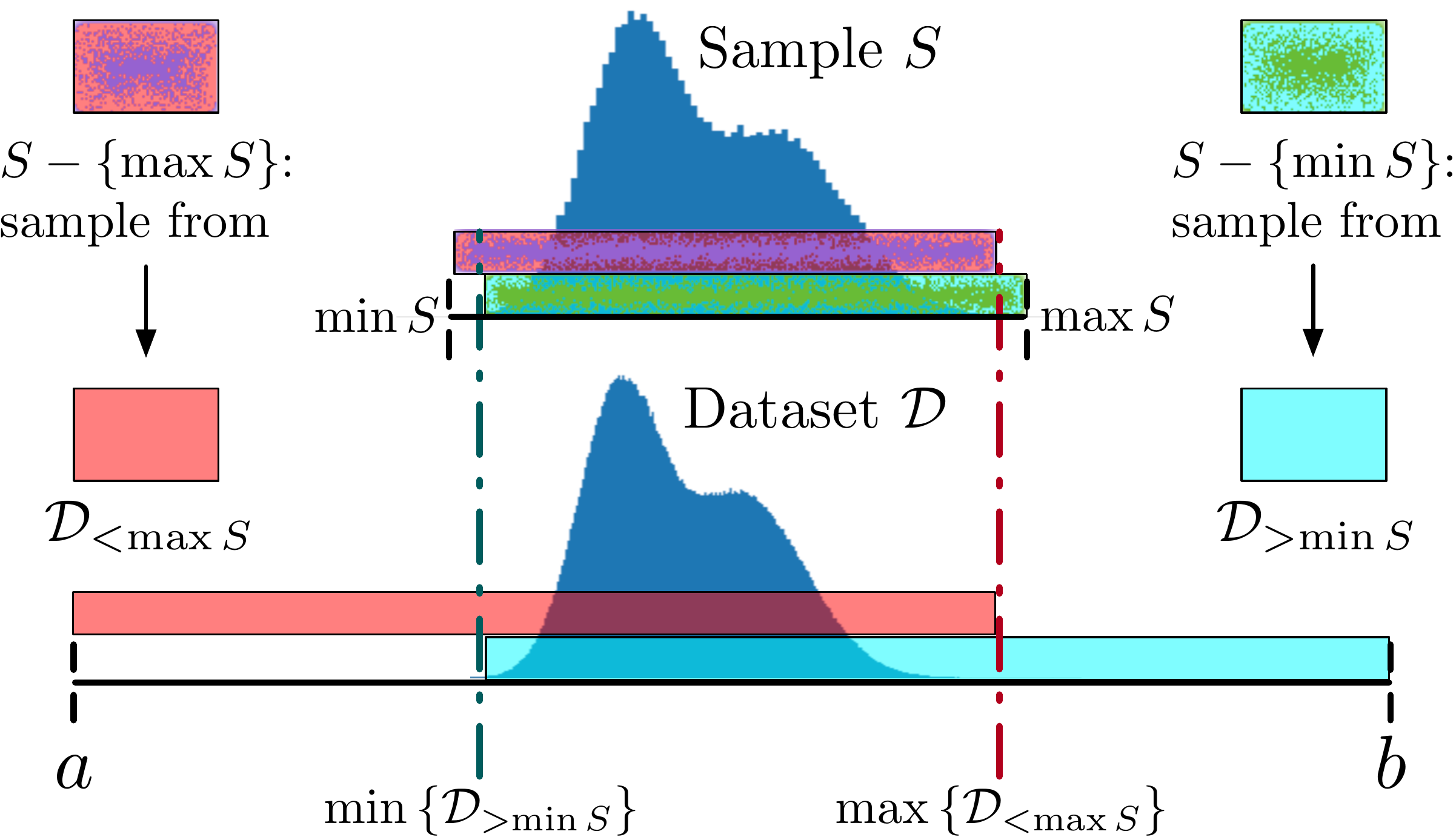}
\caption{Range trimming eliminates \pos for range-based error bounders.}
\label{fig:rangetrimming:dd}
\end{densecenter}
\end{figure}

\begin{algorithm}[t]
\caption{The \rangetrim meta-algorithm}\label{alg:rangetrim:dd}
\algrenewcomment[1]{\texttt{/*} #1 \texttt{*/} \\}
\SetKwInOut{Input}{Input}
\SetKwInOut{Output}{Output}
\SetKwRepeat{Do}{do}{while}
\SetKwBlock{Repeat}{Repeat}{}
\SetKwProg{myfunc}{function}{}{}
\SetKwFunction{sample}{sample\_without\_replacement}
\SetKwFunction{init}{init\_state}
\SetKwFunction{update}{update\_state}
\SetKwFunction{lbound}{Lbound}
\SetKwFunction{rbound}{Rbound}
{\thesisdel{\scriptsize}
\Input{Dataset $\dataset$ of $N$ values in $[a,b]$, error prob.\ $\delta$, sample size $m$}
\Output{Error bounds that fail to enclose $\AVG(\dataset)$ with probability $<\delta$}
\nonl\ \\
$S_\ell \gets \init{}$\;
$S_r \gets \init{}$\;
$a' \gets \sample{$\dataset$}$\;
$b' \gets a'$\;
\For{$\iter=1$ \KwTo $m-1$}{
	$v \gets \sample{$\dataset$}$\;
	$S_\ell \gets \update{$S_\ell, \min(v, b')$}$\;
	$S_r \gets \update{$S_r, \max(v, a')$}$\;
	$a' \gets \min(a', v)$\;
	$b' \gets \max(b', v)$\;
}
\Return{$\big[$\lbound{$S_\ell, a, b', N-1, \frac{\delta}{2}$},
\rbound{$S_r, a', b, N-1, \frac{\delta}{2}$}$\big]$}\;
}
\end{algorithm}

\subsection{Range Trimming}
Our approach to deriving an error bounder with neither \pma nor
\pos is to start with a symmetric bounder without \pma\ 
(such as \techsub{a Bernstein-based bounder}{that of \Cref{alg:ebsbounder:dd}}) and ``asymmetrize'' it
so that \lbound becomes independent of $b$, and \rbound becomes
independent of $a$, thereby eliminating \pos. The result, given
in \Cref{alg:rangetrim:dd},
\iftech
wraps any existing range-based error bounder.
\else
composes any range-based error bounder that exposes the following interface:
\begin{denseding}
{\color{linkpurple}\item\leavevmode\vphantom{a}\label{item:init}}\texttt{init\_state()}: Initializes state needed for error bounds.
{\color{linkpurple}\item\leavevmode\vphantom{b}\label{item:update}}\texttt{update\_state($S,v$)}: Given the current state $S$ and a newly-seen value $v$,
      compute new state $S'$.
{\color{linkpurple}\item\leavevmode\vphantom{c}\label{item:lbound}}\texttt{Lbound($S,a,b,N,\delta$)}: Return a confidence lower bound
              for a sample whose relevant statistics are captured in state $S$,
              assuming the sample came from a finite dataset $\dataset$ of $N$ values in $[a,b]$.
              The probability that the sample leads to this function returning a value greater
              than $\AVG(\dataset)$ is $< \delta$.
{\color{linkpurple}\item\leavevmode\vphantom{d}\label{item:rbound}}\texttt{Rbound($S,a,b,N,\delta$)}: Symmetric to \texttt{Lbound} for the confidence upper bound.
      Can typically be implemented in terms of \texttt{Lbound} after
      a suitable transformation of $S$.
\end{denseding}
The state $S$ captures information such as the count of tuples examined
and the current running average, as well as anything else required by
\texttt{Lbound} and \texttt{Rbound}. The state initialization and %
update logic is analogous to state maintenance logic for aggregates functions as
implemented in existing commercial database systems\techsub{~\cite{psqlagg}}{~\cite{mssqlagg,oracleagg,psqlagg}}.
\fi

Besides the memory required to maintain state for the left and right
error bounders, $S_\ell$ and $S_r$, \Cref{alg:rangetrim:dd}
requires $\bigo{1}$
extra memory to maintain the \MIN and \MAX element seen so far
(which replace $a$ and $b$ when computing \rbound and \lbound,
respectively).

When $\dataset$ contains unique elements, \Cref{alg:rangetrim:dd} conceptually
performs the following steps:
\begin{denseenum}
\item Sample $S$ without replacement from $\dataset$.
\item Use \lbound to compute a $1-\frac{\delta}{2}$ lower confidence bound
      for $\AVG(\dslt{\max{S}})$, with $S-\{\max{S}\}$ as the sample,
      and with $a$ and $\max{S}$ in place of the normal range bounds $a$ and $b$,
      respectively.
\item Use \rbound to compute a $1-\frac{\delta}{2}$ upper confidence bound
      for $\AVG(\dsgt{\min{S}})$, with $S-\{\min{S}\}$ as the sample,
      and with $\min{S}$ substituted for the range bound lower bound $a$.
\end{denseenum}
Note that we use $\dslt{x}$ and $\dsgt{x}$ as shorthand for $\dataset \cap (-\infty, x)$
and $\dataset \cap (x, \infty)$, respectively.
The primary difference between these high-level steps and the pseudocode presented
in \Cref{alg:rangetrim:dd} is that \Cref{alg:rangetrim:dd} maintains $\min{S}$ and $\max{S}$
in an online, streaming fashion
(so that the sample $S$ does not need to be stored in memory), and
that the confidence interval returned by \Cref{alg:rangetrim:dd} is valid
even when $\dataset$ contains duplicates (although the returned confidence bounds
will bound the \AVG of sets that differ slightly from $\dslt{\max{S}}$ and $\dsgt{\min{S}}$).
That said, we restrict our discussion and analysis to the case where $\dataset$ contains
unique elements, for simplicity.

Correctness of \Cref{alg:rangetrim:dd} crucially depends on the fact that,
conditioned on the value of $\max{S}$ (and for any such value), the
remaining elements in $S$ (namely $S-\{\max{S}\}$) constitute a uniform
without-replacement sample from $\dslt{\max{S}}$, with a symmetric
statement for $\min{S}$ and $S-\{\min{S}\}$. At a high level, this
means that a confidence lower bound computed over $S-\{\max{S}\}$
is a valid confidence lower bound for $\AVG(\dslt{\max{S}})$,
and since $\AVG(\dslt{\max{S}}) \leq \AVG(\dataset)$, it is also
a valid confidence lower bound for $\AVG(\dataset)$, with symmetric
statements holding for the confidence upper bound, $S-\{\min{S}\}$,
and $\dsgt{\min{S}}$. These core ideas are illustrated in \Cref{fig:rangetrimming:dd}.

\subsection{Proof of Correctness}
In this section, we prove correctness of \Cref{alg:rangetrim:dd} (that is,
that it returns intervals that fail to enclose $\AVG(\dataset)$ with
probability less than $\delta$).
For the sake of simplicity, our analysis assumes that $\dataset$
contains no duplicate values,
although
\iftech
we show how to remove this assumption at the end of this section.
\else
it is possible to remove this assumption.
\fi
To begin, we first prove a crucial lemma about the sampling distribution
of $S-\{\max{S}\}$, given that $S$ was sampled uniformly without-replacement
from $\dataset$.

\begin{lemma}
\label{lem:rangetrim:dd}
Given a dataset $\dataset$ of $N$ unique real values in $[a, b]$ and a uniform without-replacement
sample $S$ of $m$ values from $\dataset$, if we denote $b'=\max{S}$,
the set $S - \{b'\}$ takes the distribution
of a uniform without-replacement sample from $\dslt{b'} = \dataset \cap [a,b')$,
for any applicable value of $b'\in\dataset$.
\end{lemma}
\begin{proof}
Because $S$ is drawn uniformly without-replacement from $\dataset$, any particular instance
satisfies
\[
\Psubarg{\dataset}{S=s} = \binom{|\dataset|}{|s|}^{-1}\indic{s\subseteq\dataset} = \binom{N}{m}^{-1}\indic{s\subseteq\dataset}
\]
where we use the notation $\Psubarg{\dataset}{S=s}$ to denote the probability that
$s$ was drawn uniformly without-replacement from $\dataset$, and $\indic{\cdot}$
denotes the indicator function.
We need to show that, for any $b' \in \dataset$,
\[
\Psubarg{\dataset}{S=s|\max{S}=b'} = \Psubarg{\dslt{b'}}{S=s-\{b'\}}\indic{\max(s)=b'}
\]
First, letting $s'$ be any set such that $|s'|=m-1$, we have that
\[
\Psubarg{\dslt{b'}}{S=s'} = \binom{|\dslt{b'}|}{m-1}^{-1}\indic{s'\subseteq\dslt{b'}}
\]
Next, consider $\Psubarg{\dataset}{S=s|\max{S}=b'}$. Bayes' rule gives that
\[
\Psubarg{\dataset}{S=s|\max{S}=b'} = \frac{\Psubarg{\dataset}{S=s \wedge \max{S}=b'}}{\Psubarg{\dataset}{\max{S}=b'}}
\]
We have $\Psubarg{\dataset}{S=s \wedge \max{S}=b'} = \Psubarg{\dataset}{S=s}\indic{\max(s)=b'}$
which is a known quantity, so the key is to compute the denominator
$\Psubarg{\dataset}{\max{S}=b'}$. Using the assumption that $\dataset$ contains unique elements,
we may proceed by analogy with binary strings. The rank of $b'$ within $\dataset$ (starting from
the smallest element) is $1+|\dslt{b'}|$, so we need to compute the number of binary strings
of length $N$ containing $m$ 1's and $(N-m)$ 0's such that position $1+|\dslt{b'}|$ has a 1,
and the remaining $(m-1)$ 1's are all at positions less than $1+|\dslt{b'}|$. This is precisely
the same as the number of binary strings of length $|\dslt{b'}|$ with $(m-1)$ 1's and
$(|\dslt{b'}|-m+1)$ 0's. Putting everything together, %
\begin{align*}
&\Psubarg{\dataset}{S=s|\max{S}=b'} = \frac{\Psubarg{\dataset}{S=s}\indic{\max(s)=b'}}{\Psubarg{\dataset}{\max{S}=b'}} \\
&\quad\quad\quad=\quad \frac{\binom{N}{m}^{-1}\indic{s\subseteq\dataset\wedge\max(s)=b'}}{\binom{|\dslt{b'}|}{m-1} / \binom{N}{m}} \\
&\quad\quad\quad=\quad\binom{|\dslt{b'}|}{m-1}^{-1}\indic{s\subseteq\dataset\wedge\max(s)=b'} \\
&\quad\quad\quad=\quad\binom{|\dslt{b'}|}{m-1}^{-1}\indic{s-\{b'\}\subseteq\dslt{b'}\wedge\max(s)=b'} \\
&\quad\quad\quad=\quad\Psubarg{\dslt{b'}}{S=s-\{b'\}}\indic{\max(s)=b'}
\end{align*}
which is precisely what we wanted to show.
\end{proof}

\topic{Wrinkle in \Cref{lem:rangetrim:dd} and Fix}
The proof of \Cref{lem:rangetrim:dd} assumes unique values;
\thesistechdel{\papertext{please see our extended technical report~\cite{techreport-dd}
for how to remove this assumption without loss of generality.}}
\thesisortech{we show here how to remove this assumption without loss of generality.
The uniqueness assumption
as used is necessary only to ensure that elements of $\dataset$ are
{\em totally ordered} under some relation ``$\prec$''
(with ``$\prec$'' $\equiv$ ``$<$'' in the proof). To fix, we can simply
augment every $v\in\dataset$ with an additional unique {\em label}
(where the set of labels are totally ordered) such
that item $v$ becomes $v' \equiv (v, v_i)$. Then, define ``$\prec$''
as a relation such that $v' \prec w'$ if $v < w$, or $v=w$ and $v_i < w_i$.
In this way, any $v', w' \in \dataset'$ satisfy exactly one of $v' \prec w'$
or $w' \prec v'$, and the proof of \Cref{lem:rangetrim:dd} goes through,
replacing $\dataset$ with $\dataset'$ and
``$<$'' with ``$\prec$'' where appropriate.}

We next give a symmetric statement for $S-\{\min{S}\}$ and $\dsgt{\min{S}}$
as the below corollary:
\begin{corollary}
\label{cor:rangetrim:dd}
Given a dataset $\dataset$ of $N$ unique real values in $[a, b]$ and a uniform without-replacement
sample $S$ of $m$ values from $\dataset$ such that $\min{S}=a'$, the set $S - \{a'\}$ is a uniform
without-replacement sample from $\dsgt{a'} = \dataset \cap (a',b]$,
for any applicable value of $a'\in\dataset$.
\end{corollary}

\thesisortech{\topic{Monotonicity Property and Correctness Proof}
Before proving the main result, we briefly describe
the {\em dataset size monotonicity property} obeyed by all
bounders in this paper. THis fact will be used in the
main correctness proof.
When $N$ is unknown, an upper bound on $N$ suffices, because
bounders in this paper all satisfy the following: for any $S$,
$a$, $b$, $N$, $\delta$, and $N' > N$,
\[ \lbound(S, a, b, N', \delta) \leq \lbound(S, a, b, N, \delta) \]
\[ \rbound(S, a, b, N', \delta) \geq \rbound(S, a, b, N, \delta) \]
That is, using an upper bound for $N$ can only make the CI looser,
and since \ssi range-based error bounders with the correct
dataset size $N$ fail with probability at most $\delta$, they must
also fail with probability at most $\delta$ for any $N' > N$.}

\thesisortech{We are now ready to prove correctness of \Cref{alg:rangetrim:dd}.}
\thesistechdel{\papertext{We now give a theorem on the correctness of \Cref{alg:rangetrim:dd}.}}
\begin{theorem}
\label{thm:rangetrim:dd}
Given \ssi range-based bounders \lbound and \rbound
for computing lower (\resp upper) confidence bounds
and a dataset $\dataset$ of $N$ unique values
known to all fall in the interval $[a, b] \subseteq \reals$,
\Cref{alg:rangetrim:dd} returns a $(1-\delta)$ confidence interval for $\AVG(\dataset)$.
\end{theorem}
\begin{proof}
\thesistechdel{\papertext{Please see our extended technical report~\cite{techreport-dd}.}}
\thesisortech{\Cref{alg:rangetrim:dd} proceeds by drawing $S$ uniformly and without replacement
from $\dataset$ and computing $a'=\min{S}$, $b'=\max{S}$, $S_\ell$, and $S_r$,
where the latter two quantities capture relevant
statistics from the sample $S-\{b'\}$ and $S-\{a'\}$, respectively,
so we treat $S_\ell$ and $S_r$ as if $S_\ell = S-\{b'\}$ and $S_r = S-\{a'\}$.
By \Cref{lem:rangetrim:dd}, we have that $S_\ell$ is a uniform sample of $m-1$ values drawn without
replacement from $\dslt{b'}$, and likewise by \Cref{cor:rangetrim:dd} $S_r$ is a uniform sample
of $m-1$ values drawn without replacement from $\dsgt{a'}$. Because \lbound and \rbound
are assumed to be \ssi, range-based error bounders, we have that
\begin{align}
&\ \Parg{\lbound(S_\ell, a, b', N-1, \frac{\delta}{2}) > \AVG(\dataset)} \label{eq:supplied} \\
\leq&\ \Parg{\lbound(S_\ell, a, b', |\dslt{b'}|, \frac{\delta}{2}) > \AVG(\dataset)} \label{eq:monotone} \\
\leq&\ \Parg{\lbound(S_\ell, a, b', |\dslt{b'}|, \frac{\delta}{2}) > \AVG(\dslt{b'})} < \frac{\delta}{2}
\label{eq:trimmed}
\end{align}
and symmetrically for $\rbound(S_r, a', b, N-1, \frac{\delta}{2})$, but with ``$>$''
replaced with ``$<$'' in the probability expression above, and replacing
$\dslt{b'}$ with $\dsgt{a'}$.
$(\ref{eq:supplied})\rightarrow(\ref{eq:monotone})$ follows from the dataset size
monotonicity property of \lbound\ (\S\ref{sec:state:dd}), applicable
since $N-1 \geq |\dslt{b'}|$, and $(\ref{eq:monotone})\rightarrow(\ref{eq:trimmed})$
follows since $\AVG(\dslt{b'}) \leq \AVG(\dataset)$, as the former is clipped above $b'$
(and similarly for \rbound since $\AVG(\dsgt{a'})\geq\AVG(\dataset)$).
Union bounding over the cases for each of \lbound and \rbound,
the probability that \Cref{alg:rangetrim:dd} returns
an interval that does not enclose $\AVG(\dataset)$ is at most $\delta$.}
\end{proof}

\section{System Considerations}
\label{sec:system:dd}

In this section, we address a number
of implementation issues that become pertinent when applying
techniques of previous sections in a real system.
Although the techniques presented in this section are auxiliary to our
primary contribution and can be used with any CI approach, they are developed
with \ssi error bounders and strong probabilistic guarantees in mind.
First, we describe how to augment the techniques of \Cref{sec:algorithm:dd},
which apply for a fixed sample size taken without replacement from
a finite dataset of known size, with locality-aware
{\em scan-based} without-replacement sampling that
need not know $N$, and we further describe how to use this
layout to facilitate \SUM and \COUNT aggregations
(\S\ref{sec:scansample:dd}).
Next, we describe
an {\em optional stopping} routine that does not require
a sample size to be specified up-front (\S\ref{sec:optstop:dd}).
Finally, we describe
an {\em active scanning} architectural optimization that prioritizes
samples that facilitate early termination (\S\ref{sec:active:dd}),
all without losing guarantees proved in
\Cref{sec:algorithm:dd}.

These system details are implemented within the context
of \fv, which is our general relational column store 
for approximate
report generation with
guarantees. \fv uses the error bounders from \Cref{sec:algorithm:dd}
and pairs them with a practical architecture for without-replacement sampling.
\fv uses block-based bitmaps over categorical
attributes (similar to~\cite{macke2018adaptive})
for efficient processing of queries with predicates or groups.
Furthermore, for continuous attributes, \fv stores the minimum
and maximum values in a catalog, to be used as the range bounds
$a$ and $b$ for the desired range-based error bounder.

\subsection{Scan-Based Sampling for DB Aggregates}
\label{sec:scansample:dd}
We now describe how \fv implements without-replacement
sampling in a locality-aware manner by {\em scanning} over pre-shuffled
data, and furthermore how this approach can be
used to compute CIs for \COUNT and \SUM.
The up-front shuffling cost need only be paid
once in order to facilitate many queries, although care must be taken
to set the error probability $\delta$ small enough when running multiple queries
to avoid losing error bounder guarantees.
The high-level idea behind this implementation of without-replacement
sampling is not new and has been used in prior work for approximate query
processing~\cite{qin2014pf,wu2010continuous,zeng2015g,macke2018adaptive}
and other analytics~\cite{bismarck}.
We begin by introducing
{\em scrambles} and {\em aggregate views}:

\smacke{Figure for this?}

\begin{definition}[Scramble]
\label{def:scramble}
A {\em scramble} is an {\em ordered} copy of a relational table
that has been permuted randomly,
allowing for scan-based without-replacement sampling.
\end{definition}

Scanning a continuous column in a scramble is equivalent to sampling
without replacement. In fact, scanning any subset of data in
a continuous column in a scramble (assuming the subset is chosen
without knowledge of the order of data)
is also equivalent to sampling without replacement,
so that scanning a scramble can be used to sample without replacement
for any aggregate appearing in a query containing arbitrary filters
or \groupby clauses. We call such subsets {\em aggregate views}:

\begin{definition}[Aggregate View]
\label{def:sampview}
An {\em aggregate view} for some aggregate $A$ appearing
in a query (possibly belonging to a group induced by a \groupby clause)
is the set of values in a scramble that contribute toward 
the computation of $A$.
\end{definition}

Note that $\delta$ must be divided by the number of aggregate views
in a query (or an upper bound) to preserve error guarantees.

\topic{Computing CIs for \COUNT}
Ensuring that data in a scramble are permuted randomly makes it
easy to compute bounds on the selectivities of aggregate views,
and by extension on the \COUNT of tuples in each aggregate view,
using existing techniques~\cite{haas1997large,haas1999ripple}.
\thesistechdel{\papertext{Please see the full discussion in our extended technical
report~\cite{techreport-dd}; we briefly outline the main idea here.}}
\thesistext{We briefly outline the main idea.}
One can conceptually assign each row of a scramble a $1$ if it belongs
to the aggregate view of interest, and a $0$ otherwise. The \AVG of
this ``derived'' view (over the whole scramble) is exactly the selectivity
of the aggregate view, and we can use a Hoeffding-Serfling-based bounder
to compute a CI for the selectivity (using range bounds of $a=0$ and $b=1$).
Multiplying these bounds by the total
number of rows in the scramble then yields a CI for the \COUNT of rows
that participate in the aggregate view. \smacke{Argue that HS is appropriate?}

\papertext{As an aside, recall that the range-based error bounders we consider in \Cref{sec:experiments:dd}
all take as input the number of tuples in an aggregate view $N_i$.
In fact, an upper bound on $N_i$ suffices to preserve correctness.
Thus, the same method for computing CIs for \COUNT aggregates
can be used in conjunction with error bounders for \AVG by using
an upper bound on \COUNT in place of the exact dataset size $N_i$\thesistechsub{;}{.}}
\thesistechdel{\papertext{please see our extended technical report~\cite{techreport-dd} for details.}}
\techreport{In more detail, %
for a scramble with $R$ rows, $N$ of which are in the sample view $V$
for a query $Q$, the number of rows seen that belong to $V$, $m_v$, after
scanning $r$ rows of the scramble is a hypergeometric random variable~\cite{casella2002statistical}
whose mean is the selectivity $\sigma_v$ of $V$ multiplied by $r$, $\sigma_v\cdot r$. One could use
bounds specifically tailored to the hypergeometric distribution (or even perform an exact
computation) to compute an upper bound on $\sigma_v$ that holds \whp, but in this work we use
a simple strategy that uses Hoeffding-Serfling to bound $\sigma_v$, stated as follows.

\begin{lemma}
\label{lem:sel-ubound:dd}
The probability that a scan of a scramble of size $R$ that has processed
$r$ rows so far sees fewer than $(\sigma_v-\veps)\cdot r$
or more than $(\sigma_v+\veps)\cdot r$ rows belonging
to $V$ is at most $\delta$, for $\veps = \sqrt{\frac{\log(2/\delta)}{2r}\cdot(1-\frac{r-1}{R})}$.
\end{lemma}
\begin{proof}
Follows immediately from application of the Hoeffding-Serfling
inequality~\cite{serfling1974probability}.
\end{proof}
\Cref{lem:sel-ubound:dd} implies that, for a scan that has seen $m_v$
rows so far belonging to $V$, $\sigma_v$ is within
\[ \whs_v \pm \veps = \frac{m_v}{r} \pm \sqrt{\frac{\log(2/\delta)}{2r}\cdot\left(1-\frac{r-1}{R}\right)} \]
\whp. This in turn implies that $N$, the number of tuples belonging to $V$,
is within $[N^-, N^+] = [(\whs_v-\veps)\cdot R, (\whs_v+\veps)\cdot R]$, \whp.

\topic{Combining \Cref{lem:sel-ubound:dd} with error bounders}
For error bounders \lbound and \rbound of the form described in
\Cref{sec:algorithm:dd} that require the data range bounds $a$ and $b$
as well as the data size $N$, a $(1-\delta)$ confidence interval
is computed as
\[ \left[\lbound(S, a, b, N, \delta/2),\ \rbound(S, a, b, N, \delta/2)\right] \]

The following theorem describes how to use \Cref{lem:sel-ubound:dd}
to compute $(1-\delta)$ error bounds when $N$ is unknown.

\begin{theorem}
\label{thm:unknown-size}
Consider a query $Q$ operating over some size-$R$ scramble with corresponding
sample view $V$.
Suppose a scan the scramble
that has processed $k$ rows so far has seen $m_v$ rows belonging
to $V$ (from which $S$ is computed). Letting
\[ N^+ = \left(\frac{m_v}{r} + \sqrt{\frac{\log(1/(1-\alpha)\cdot\delta)}{2r}\cdot\left(1-\frac{r-1}{R}\right)}\right)\cdot R \]
then the interval
\[ \left[\lbound(S, a, b, N^+, \alpha\cdot\delta/2),\ 
         \rbound(S, a, b, N^+, \alpha\cdot\delta/2) \right] \]
is a $(1-\delta)$ confidence interval for the mean of $V$, for any $\alpha \in (0,1)$.
\end{theorem}
\begin{proof}
Conditioning over whether $N\leq N^+$ or $N>N^+$,
the probability that the aforementioned
interval $[L,R]$ fails to contain the desired mean $\mu$ is
\[ \Parg{\mu\notin[L,R]\ |\ N > N^+}\cdot\Parg{N > N^+} \]
\[ + \Parg{\mu\notin[L,R]\ |\ N\leq N^+}\cdot\Parg{N\leq N^+} \]
\[ \leq \Parg{N > N^+} + \Parg{\mu\notin[L,R]\ |\ N\leq N^+} \]
By \Cref{lem:sel-ubound:dd},
the first probability is at most $(1-\alpha)\cdot\delta$, and \lbound and \rbound
are such that the probability of the second term is at most $\alpha\cdot\delta$.
The sum of these equals $\delta$, implying that the interval is a valid $(1-\delta)$
confidence interval, completing the proof.
\end{proof}
Throughout \Cref{sec:experiments:dd}, we fix $\alpha=0.99$, giving most of the weight
to the confidence interval computation, corresponding to a looser upper bound for $N$.}

\topic{Computing CIs for \SUM}
Now that we have established how to compute CIs for \AVG and \COUNT,
we briefly describe how to combine these two techniques to compute CIs
for \SUM{}\thesistechdel{\papertext{ (please see the technical report~\cite{techreport-dd}
for more discussion)}}. Given a $(1-\frac{\delta}{2})$ confidence interval for \COUNT
as $[c_\ell, c_r]$ and a $(1-\frac{\delta}{2})$ confidence interval for \SUM
as $[\gell, \garr]$,
union bounding gives $[c_\ell\cdot\gell, c_r\cdot\garr]$
as a $(1-\delta)$ confidence interval for \SUM.

\subsection{Optional Stopping}
\label{sec:optstop:dd}
\techreport{\begin{algorithm}[t]
\caption{The \optstop meta-algorithm}\label{alg:optstop:dd}
\algrenewcomment[1]{\texttt{/*} #1 \texttt{*/} \\}
\SetKwInOut{Input}{Input}
\SetKwInOut{Output}{Output}
\SetKwRepeat{Do}{do}{while}
\SetKwBlock{Repeat}{Repeat}{}
\SetKwProg{myfunc}{function}{}{}
\SetKwFunction{sample}{sample\_without\_replacement}
\SetKwFunction{init}{init\_state}
\SetKwFunction{update}{update\_state}
\SetKwFunction{lbound}{Lbound}
\SetKwFunction{rbound}{Rbound}
{\thesisdel{\scriptsize}
\Input{Dataset $\dataset$ of $N$ values in $[a,b]$, error probability $\delta$}
\Output{Error bounds that fail to enclose $\AVG(\dataset)$ with probability $<\delta$}
\nonl\ \\
$S \gets \init{}$\;
\For{$k=1$ \KwTo $\infty$}{
\label{line:loop-round:dd}
	\For{$\iter=1$ \KwTo $B$}{
		$v \gets \sample{$\dataset$}$\;
		$S \gets \update{$S, v$}$\;
	}
	$\delta' \gets (6/\pi^2) \cdot (\delta / k^2)$\;
	$L_k\gets$\lbound{$S, a, b, \frac{\delta'}{2}$}\;
	$R_k\gets$\rbound{$S, a, b, \frac{\delta'}{2}$}\;
	\If{\texttt{\upshape{should\_stop}}$(%
	\big[\max_{j\leq k}\{L_j\}, \min_{j\leq k}\{R_j\}\big])$}{
\label{line:stop-cond:dd}
	\texttt{break}\;
	}
}
\Return{$\big[\max_k\{L_k\}, \min_k\{R_k\}\big]$}
}
\end{algorithm}
}
The techniques discussed in \Cref{sec:algorithm:dd} describe how to
compute high-probability bounds on error given statistics computed from a
particular sample of $m$ datapoints.
Fixing a sample size ahead of time is oftentimes impractical,
since it is usually unknown how many samples are needed to ensure
CIs that are ``just tight enough'' to facilitate downstream applications
on the part of the user or the system.
For example, one approach (\eg taken by VerdictDB~\cite{park2018verdictdb})
is to first compute error bounds around an approximate aggregate, and then
run an exact query if these bounds are too loose.

Another approach, which we take in this \work, is to continue
taking samples until a bound on the error is provably small enough.
For this approach, care must be taken to avoid losing guarantees
offered by range-based error bounders, since the tighter of two $(1-\delta)$
confidence intervals for a particular aggregate is itself not necessarily
a $(1-\delta)$ confidence interval. \papertext{As the technique employed
is orthogonal to the primary contribution of this work, which is to eliminate
\pma and \pos from range-based error bounders, we omit discussion here and
instead give a treatment in the technical report~\cite{techreport-dd}.}

\techreport{Various techniques have been developed for computing {\em sequentially-valid}
confidence intervals as new samples are
taken~\cite{wald2004sequential,zhao2016adaptive,Kim2015,lipton1993efficient,lipton1990practical,lipton1989estimating}.
In addition to techniques for
sequential estimation~\cite{wald2004sequential}
for sequences of \iid random variables from a known family of distributions,
various concentration results applicable to \AVG
which make no distributional assumptions have
likewise been derived~\cite{zhao2016adaptive,Kim2015}. Unfortunately, these
existing results are derived from variants of Hoeffding's inequality, and
therefore suffer from \pma. For the sake of simplicity, we use a much simpler
meta-algorithm that can be used in conjunction with any range-based error
bounder, including those that leverage our \rangetrim technique, given in
\Cref{alg:optstop:dd}. Although \Cref{alg:optstop:dd} requires more samples than
the aforementioned techniques when used in conjunction with Hoeffding-
or Hoeffding-Serfling-based error bounders, we consider the tradeoff worthwhile
due to its generality and simplicity and leave better sequential error bounders
to future work.

\topic{Analysis of Algorithm~\ref{alg:optstop:dd}}
\Cref{alg:optstop:dd} proceeds in ``rounds'', with each iteration of the outer
loop on \cref{line:loop-round:dd} forming a round. During each round, $B$ without-replacement
samples are taken and used to incrementally update the state of any range-based
error bounder that implements our interface from \Cref{sec:state:dd}. At the end of
each round, confidence intervals are recomputed, with the input error probability
$\delta'$ decayed ``enough'' to ensure that the probability of error across {\em all}
rounds is at most $\delta$. If the stopping condition on \cref{line:stop-cond:dd} is met, then
the algorithm terminates; otherwise, it proceeds to the next round, decaying $\delta'$
appropriately to control the overall error probability.

We now give a proof of correctness of \Cref{alg:optstop:dd}.
\begin{theorem}
With probability at least $(1-\delta)$,
the $\{L_k\}$ and $\{R_k\}$ computed by \Cref{alg:optstop:dd} satisfy
$\AVG(\dataset) \in [L_k, R_k]$ for every $k$ in the outer loop.
In particular,
$\AVG(\dataset) \in [\max_{k\geq1}{L_k}, \min_{k\geq1}{R_k}]$
with probability at least $(1-\delta)$.
\end{theorem}
\begin{proof}
Denote the $\delta'$ used at iteration $k$ with $\delta_k$. Union bounding
over rounds, the probability of a mistake is at most
\[ \sum_{k\geq1}\delta_k = \frac{6}{\pi^2}\sum_{k\geq1}\frac{\delta}{k^2} = \frac{6}{\pi^2}\cdot\frac{\pi^2}{6}\delta = \delta\]
using the identity $\sum_{k\geq1}\frac{1}{k^2} = \frac{\pi^2}{6}$,
completing the proof.
\end{proof}

\fv performs I/O at the level of blocks,
so instead of computing bounds every $B$ samples
as described in the pseudocode of \Cref{alg:optstop:dd}, we compute bounds
after every $B$ block read. In our experiments
(\Cref{sec:experiments:dd}), we set $B=40000$.
We leave %
development of alternative approaches
to future work.}

\topic{Stopping Conditions\techreport{ for Algorithm~\ref{alg:optstop:dd}}}
\techreport{Correctness of \Cref{alg:optstop:dd} is independent of whether the error bounder uses
our \rangetrim technique (please see \techsub{our extended technical
report~\cite{techreport-dd}}{\Cref{sec:rtbounder:dd} in the appendex} for an implementation
of our \rangetrim technique in terms of the interface from \Cref{sec:state:dd}), and
it is furthermore independent of stopping condition.}
We consider several stopping
conditions used in our system implementation:
\begin{denseding}
\iftech
\item \textbf{Desired Samples Taken} (\enough): If a fixed number of
	samples are requested, do not use \Cref{alg:optstop:dd}; instead,
	terminate query processing once a desired number of tuples contribute
	to the partial aggregate(s) in the query.
	\label{stop:taken}
	\item \textbf{Sufficient Absolute Accuracy} (\intervalsmallerthan{\veps}):
	The interval width is sufficiently small.
	\label{stop:absacc}
\fi
	\item \textbf{Sufficient Relative Accuracy} (\relintervalsmallerthan{\veps}):
	The interval width is sufficiently small
	(relative to the possible correct values implied by the interval).
	\label{stop:relacc}
	\item \textbf{Threshold Side Determined} (\notinside{v}):
	The interval does not contain some threshold value $v$,
	indicating that the true \AVG is \whp either less than
	or greater than the threshold $v$.
	\label{stop:thresh}
	\item \textbf{Top- or Bottom-$K$ Separated}: In a query
	with multiple groups, the error bounds of the groups with either $K$ smallest or largest
	aggregates do not intersect those of any of the remaining groups.
	\label{stop:sep}
	\item \textbf{Groups Ordered Correctly}: In a query
	with multiple groups, the error bounds for each group intersect none of the
	other groups' error bounds, indicating that the correct ordering of group
	aggregates has been determined~\cite{Kim2015}.
	\label{stop:ordered}
\end{denseding}
Different stopping conditions apply to different queries.
For example, stopping conditions~\ref{stop:relacc} and~\ref{stop:thresh}
might be used for the query in \Cref{fig:putative:dd}.

\subsection{Active Scanning}
\label{sec:active:dd}
For queries with \sql{GROUP BY}s, different groups may require different
numbers of samples to achieve stopping conditions of the types considered
in \Cref{sec:optstop:dd}. For simple scans that simply read blocks of the
scramble in the order in which they appear, it is impossible to control
the relative number of tuples for each group, leading to potential
inefficiencies. For example, consider one of the queries in our
experiments, \fqref{airlinepos}, which selects airlines with
average delay above some threshold. This query uses stopping condition~\ref{stop:thresh} in
order to determine when to terminate, since, when this stopping condition
has been achieved, it as been determined \whp whether each airline has average
delay above or below the threshold.  Those groups (airlines) for which the average delay
is near \blue{\texttt{\$thresh}} require more samples than those for which the average delay
is far from \blue{\texttt{\$thresh}} in order to achieve condition~\ref{stop:thresh}. If these
groups are sparse within the scramble, a scan will look at much
more data than necessary.

For this reason, we process queries that perform \sql{GROUP BY}s with an
adaptive sampling approach using
{\em active scanning}, which is a block-skipping technique
that only processes blocks that contain tuples for so-called {\em active groups},
skipping any other blocks.
The notion of an active group depends on the stopping condition,
but in brief, active groups are groups that should be prioritized
for sampling to more quickly achieve a corresponding stopping condition.
\thesistechdel{\papertext{Please see our extended technical report~\cite{techreport-dd}
for a description of active groups for each applicable stopping condition.}}

\iftech
\topic{Active Groups for Stopping Conditions}
We now describe how we determine active groups, or groups that should be prioritized for
sampling, for each of the stopping conditions discussed in \S\ref{sec:optstop:dd}.
\\ \ref{stop:taken} (Desired Samples Taken):
      Under this condition, we consider a group active
      as long as fewer than the desired $m$ samples have been taken that contribute to
      that group's aggregate.
\\ \ref{stop:absacc} (Sufficient Absolute Accuracy):
	  We consider a group active as long as its confidence bounds exceed $\veps$
	  in width; \ie, \intervalgeq{\veps}.
\\ \ref{stop:relacc} (Sufficient Relative Accuracy):
	  Same as the previous, but a group is active if \relintervalgeq{\veps}.
\\ \ref{stop:thresh} (Threshold Side Determined):
	  A group is active as long as the threshold side has not been
	  determined; \ie, \isinside{v}.
\\ \ref{stop:sep} (Top- or Bottom-$K$ Separated):
	 The activeness condition for this stopping condition is
	 the most complicated out of any we consider.
	 First, consider the sorting the aggregates for all the groups in increasing order. A group
	 among the top-$K$ is active if its lower confidence bound $\gell$ crosses the midpoint
	 between the aggregate value for the smallest of the top-$K$ and the largest of the bottom
	 $N-K$.
	 Likewise, a group among the
	 bottom $N-K$ is active if its upper confidence bound $\garr$ crosses this midpoint.
	 Analogous statements hold if we consider the separation between the bottom-$K$ and the top $N-K$
	 as the stopping condition, but with
	 the aggregates sorted in decreasing order.
\\ \ref{stop:ordered} (Groups Ordered Correctly):
	 A group is active if its interval $\interval$ intersects the interval of any other group.

\topic{Async Lookahead}
\fi
We furthermore accelerate active scanning with an
asynchronous lookahead technique from prior work~\cite{macke2018adaptive},
which we briefly describe here.
Active scanning with lookahead uses block-based bitmap indexes to efficiently check whether a block
contains tuples for any active group. Instead of synchronously checking
whether a given block contains tuples for active groups by iterating over
each active group and querying the index, a separate {\em lookahead thread}
iterates over a batch of $1024$ blocks and marks blocks for processing or skipping,
for each group. By iterating over an entire batch of $1024$ blocks for a given active group,
bitmaps for the group tend to be in cache more often, making the index
lookup operation more efficient. We refer the reader to~\cite{macke2018adaptive} for more details.
In our experiments in \Cref{sec:experiments:dd},
we set the block size to $25$ rows, so a batch of $1024$ blocks contains
a total of $25600$ rows.

\section{Empirical Study}
\label{sec:experiments:dd}
In this section, we perform an extensive empirical evaluation
of various error bounders and sampling strategies on real data.
\techreport{
\begin{table}[t]
\begin{densecenter}
{\papertext{\scriptsize}\begin{tabular}{ | c || c|c|c|c|} \hline
  \rowcolor[HTML]{C0C0C0} 
  \textbf{Dataset} & Size & \#Tuples & \#Attributes & Replications\\ \hline \hline
  \flights & 32 GiB & \numflightstuples  & 5  & $5\times$ \\ \hline
\end{tabular}}
\caption{Descriptions of \flights Dataset \smacke{Revisit this}} %
\label{tab:datasets:dd}
\end{densecenter}
\end{table}
}

\begin{table}[t]
\begin{densecenter}
{\papertext{\scriptsize}\begin{tabular}{|$c||^c^c|^C{\thesissub{\techsub{3.2cm}{2.55cm}}{7.5cm}}|}\hline
\rowcolor[HTML]{C0C0C0}\rowstyle{\bfseries}
Query                              & \multicolumn{2}{^c|}{Stop When}             & Parameters Varied \\\hline\hline
\fqref{ord}                & (\ref{stop:relacc})  & \relintervalsmallerthan{\blue{\veps}}       &
\blue{\$airport}~(\Cref{fig:vary-sel:dd}), \blue{$\veps$}~(\Cref{fig:vary-eps:dd})\\\hline
\fqref{airlinepos}           & (\ref{stop:thresh})  & \notinside{\blue{\text{\$thresh}}} & \blue{\$thresh}~(\Cref{fig:vary-having:dd})  \\\hline
\fqref{minlatedelay}               & (\ref{stop:sep})     & \bottomseparated{2}  &   \blue{\$min\_dep\_time}\papertext{~}\techreport{\hfill}(\Cref{fig:vary-deptime:dd})      \\\hline
\iftech
\fqref{ordten}                     & (\ref{stop:thresh})  & \notinside{10}       &    \na     \\\hline
\fqref{negdelay}                   & (\ref{stop:thresh})  & \notinside{0}        &    \na     \\\hline
\fqref{multigroup}                 & (\ref{stop:sep})   & \topseparated{5} &  \na \\\hline
\fqref{hpweekbar}                  & (\ref{stop:ordered}) & \ordered             &   \na      \\\hline
\fqref{topkdelay}                  & (\ref{stop:sep})     & \topseparated{1}     &   \na      \\\hline
\fqref{airlinemax}                 & (\ref{stop:sep})     & \topseparated{1}     &   \na      \\\hline
\else
\fi
\end{tabular}}
\caption{\papertext{\scriptsize}Summary of stopping conditions used for queries
provided in \Cref{fig:flightsqueries:dd}.
Template variable arguments shown in \blue{blue}.}
\label{tab:qtemplates:dd}
\undefnonheader
\end{densecenter}
\end{table}

\subsection{Flights Dataset and Queries}
\iftech
We evaluate various error bounding techniques on the
publicly available \flights dataset~\cite{flights}, extracting attributes for
origin airport, airline, departure delay, departure time, and day of week.
The details for this dataset are summarized in \Cref{tab:datasets:dd}.
The replication value indicates how many times the dataset was
replicated to create a larger dataset and ensure sufficient scale for our experiments.
We eliminated rows with ``N/A'' or erroneous values for  any column appearing
in one or more of our queries.
\else
We evaluate various error bounding techniques on the
publicly available \flights dataset~\cite{flights}, and extract
five attributes corresponding to the
origin airport, airline, departure delay, departure time, and day of week.
To ensure sufficient scale of the data, we perform $5$ replications,
giving a $32$ GiB dataset of $606$ million tuples in total.
\fi

\techsub{\begin{figure}[t]}{\begin{figure*}[t]}
\begin{densecenter}
\techreport{\begin{multicols}{2}}

\begin{flightsquery}[ord]{avg delay for \blue{\$airport}}
\begin{lstsql}
<@\qheader{}@>
SELECT AVG(DepDelay) FROM flights WHERE Origin = <@\blue{\$airport}@>
\end{lstsql}
\end{flightsquery}

\begin{flightsquery}[airlinepos]{airlines with avg delay above \blue{\$thresh}}
\begin{lstsql}
<@\qheader{}@>
SELECT Airline FROM flights
GROUP BY Airline HAVING AVG(DepDelay) > <@\blue{\Small\$thresh}@>
\end{lstsql}
\end{flightsquery}

\begin{flightsquery}[minlatedelay]{2 airlines with min avg delay after \blue{\$min\_dep\_time}}
\begin{lstsql}
<@\qheader{}@>
SELECT Airline FROM flights WHERE DepTime > <@\blue{\Small\$min\_dep\_time}@>
GROUP BY Airline ORDER BY AVG(DepDelay) ASC LIMIT 2
\end{lstsql}
\end{flightsquery}

\iftech
\begin{flightsquery}[ordten]{whether ORD has avg delay $>10$}
\begin{lstsql}
<@\qheader{}@>
SELECT (CASE WHEN AVG(DepDelay) > 10 THEN 1 ELSE 0 END)
FROM flights WHERE Origin = 'ORD'
\end{lstsql}
\end{flightsquery}

\begin{flightsquery}[negdelay]{airports with negative avg departure delay}
\begin{lstsql}
<@\qheader{}@>
SELECT Origin FROM flights
GROUP BY Origin HAVING AVG(DepDelay) < 0
\end{lstsql}
\end{flightsquery}

\begin{flightsquery}[multigroup]{5 worst days for afternoon delays across airports}
\begin{lstsql}
<@\qheader{}@>
SELECT DayOfWeek, Origin FROM flights
WHERE DepTime > <@\textcolor{codegreen}{1:50pm}@> GROUP BY DayOfWeek, Origin
ORDER BY AVG(DepDelay) DESC LIMIT 5
\end{lstsql}
\end{flightsquery}

\begin{flightsquery}[hpweekbar]{avg delay by day of week for airline HP}
\begin{lstsql}
<@\qheader{}@>
SELECT DayOfWeek, AVG(DepDelay) FROM flights
WHERE Airline = 'HP' GROUP BY DayOfWeek
\end{lstsql}
\end{flightsquery}

\begin{flightsquery}[topkdelay]{origin airport with highest departure delay}
\begin{lstsql}
<@\qheader{}@>
SELECT Origin FROM flights GROUP BY Origin
ORDER BY AVG(DepDelay) DESC LIMIT 1<@%
@>
\end{lstsql}
\end{flightsquery}

\begin{flightsquery}[airlinemax]{airline with maximum avg delay}
\begin{lstsql}
<@\qheader{}@>
SELECT Airline FROM flights GROUP BY Airline
ORDER BY AVG(DepDelay) DESC LIMIT 1
\end{lstsql}
\end{flightsquery}%
\else
\fqtechlabel{ordten}
\fqtechlabel{negdelay}
\fqtechlabel{multigroup}
\fqtechlabel{hpweekbar}
\fqtechlabel{topkdelay}
\fqtechlabel{airlinemax}
\fi
\techreport{\end{multicols}}
\papertext{\vspace{-4pt}}
\caption{SQL \techreport{and semantics} for \papertext{the first three} \flights queries. Template parameters are shown in \blue{\$blue}.}
\label{fig:flightsqueries:dd}
\end{densecenter}
\techsub{\end{figure}}{\end{figure*}}

\topic{Queries and Query Templates}
We evaluate our techniques on a diverse set of queries
that include various filters and \sql{GROUP BY} clauses
and exercise all the stopping conditions described in
\Cref{sec:optstop:dd}\techreport{ (except conditions~\ref{stop:taken}
and~\ref{stop:absacc},
which gives similar behavior to condition~\ref{stop:relacc})}.
\iftech
The queries themselves are given
in \Cref{fig:flightsqueries:dd}, and the accompanying
stopping conditions are summarized in \Cref{tab:qtemplates:dd}.
Additionally, several queries are parametrized, in order to reveal
interesting data-dependent behavior by varying corresponding parameters.
Any query parameters varied are
also summarized in \Cref{tab:qtemplates:dd}, with parameters
shown in \blue{blue}.
\else
The first $3$ of the $9$ queries in our set are presented
in full in \Cref{fig:flightsqueries:dd} and are additionally
parameterized --- by varying the corresponding parameter (shown
in \blue{blue} in \Cref{fig:flightsqueries:dd}), we reveal interesting
data-dependent behavior.
The stopping conditions for these queries are summarized in \Cref{tab:qtemplates:dd}.
Descriptions of the remaining $6$ queries are
presented in the technical report~\cite{techreport-dd} to save space.
\fi

\subsection{Experimental Setup}
The core of our experiments consists of two ablation studies,
intended to evaluate the impact of both our error bounder
innovations {\em and} that of our architectural innovations.
In particular, we evaluate various error bounders with and
without our \rangetrim technique developed in \Cref{sec:algorithm:dd},
and for the best error bounder (\bernrt), we furthermore evaluate
the impact of leaving out features of our active scanning sampling
strategy described in \Cref{sec:system:dd}.

We set $\delta=10^{-15}$ for all queries, as we expect
users of with-guarantees AQP to desire results that are
correct in an {\em effectively deterministic} manner.
Union bounding over the
number of queries run, the upper bound on the error
probability will still be sufficiently
small to guarantee correctness of all queries, for
any practical number of queries encountered.

\emtitle{Approaches.} We used the following strategies to bound error when
running queries in \Cref{fig:flightsqueries:dd}:

\topic{Error Bounders}
\begin{denselist}
\item \bernrt. This uses the empirical Bernstein-Serfling error bounder described in
\Cref{sec:pathologies:dd}, coupled with our \rangetrim technique described in \Cref{sec:algorithm:dd},
               which eliminates \pos.
\item \bern. Same as the previous, but without \rangetrim. \bern and \bernrt
             are included to evaluate the impact of an error bounder without \pma.
\item \hoefrt. This uses the Hoeffding-Serfling error bounder described in
\Cref{sec:pathologies:dd}, coupled with our \rangetrim technique described in \Cref{sec:algorithm:dd},
               which eliminates \pos from \hoef{} (but does not fix \pma).
\item \hoef. Same as the previous, but without \rangetrim.
\item \exact. This strawman approach eschews approximation and runs queries exactly,
to serve as a simple baseline.
\end{denselist}

We furthermore used the following strategies for sampling when
running queries in \Cref{fig:flightsqueries:dd}:

\topic{Sampling Strategies}
\begin{denselist}
\item \peek. This uses the active scanning technique to prioritize groups
that are preventing satisfaction of various stopping conditions, along with cache-efficient
queries to bitmaps with lookahead (see \Cref{sec:active:dd} for details).
\item \sync. This uses active scanning, but processes each block synchronously when
deciding whether to read it, incurring
high overhead since queries to bitmaps typically result in cache misses.
\item \scan. This strategy does not leverage bitmaps
in order to decide whether to read a block for active scanning
(but may leverage bitmaps for evaluation of whether a block contains
tuples that satisfy a fixed predicate, such as the one appearing in
\fqref{ord}). Without any predicate, this approach
simply processes all blocks in the scramble sequentially.
Note that the \exact baseline described previously always uses \scan,
as only approximate approaches can prune groups.
\end{denselist}

\emtitle{Environment.} Experiments were run on single Intel Xeon E5-2630 node with 125 GiB of RAM and with 8 physical cores
(16 logical) each running at 2.40 GHz, although we restrict our experiments to a single thread
(excepting the \peek sampling strategy, which uses one extra thread), noting
that our techniques can be easily parallelized.
The Level 1, Level 2, and
Level 3 CPU cache sizes are, respectively: 512 KiB, 2048 KiB, and
20480 KiB. We ran Linux with kernel version 2.6.32.
We report results for data stored in-memory,
since the cost of main memory has decreased to the point that many interactive workloads
can be performed entirely in-core.
Each approximate query was started from a random position in the shuffled data.
We found wall clock time to be
stable for all approaches, and report times
as the average of 3 runs for all methods.

\subsection{Metrics}
We gather several metrics in order to test two hypothesis:
one, that our error bounding strategies in conjunction with
our sampling strategies lead to speedups over simpler baselines;
and two, that they do so without sacrificing correctness of query
results.

\topic{Correctness of Query Results}
The most important metric is the fraction of queries run that
returned correct results. {\em Across all methods, all queries,
and all parameter settings, results either matched the ground truth
determined from an \exact evaluation, or were within error tolerance
in the case of \fqref{ord} and \fqref{hpweekbar}.} This is expected,
given that we are considering \ssi error bounders with
strong probabilistic guarantees in this \work, coupled with the fact
that our \rangetrim technique and system architecture do not compromise
these guarantees. As such, we expect fewer than $\delta=10^{-15}$
fraction of queries to yield incorrect results, which rounds down to
a cool $0$.

\topic{Estimate Error}
For a given requested error bound $\veps$ supplied to applicable queries,
we measure the actual error. The observed
error should always fall within the requested error bound.

\topic{Wall-Clock Time}
Our primary metric evaluates the end-to-end time required for
various error bounders and various sampling strategies
(where the \exact baseline is included as a ``sampling strategy''), across
all the queries considered.

\topic{Number of Blocks Fetched}
We also measure the number of blocks fetched from main memory
into CPU cache when using various approaches. This is mainly
due to the fact that error bounders incur additional CPU overhead
and therefore wall-clock time,
with \bern and \bernrt incurring the highest overhead, so measuring
blocks fetched for these approaches removes this confounding variable
by decoupling performance from CPU attributes.

\subsection{Results}
\label{sec:results}

\techreport{
\thesissub{\begin{table*}[t]}{\begin{figure}[t]}
\begin{densecenter}
{\thesistext{\small} \begin{tabular}{|l|c||c|c|c|c|} \hline
\multicolumn{1}{|c|}{\cellcolor[HTML]{C0C0C0} \textbf{Query}}
&  \cellcolor[HTML]{C0C0C0} & \multicolumn{4}{^c|}{\cellcolor[HTML]{C0C0C0}\textbf{Avg Speedup over \exact{} (raw time in (s))}} \\ \hline \hline
                 & \exact{} (s) & \hoef & \hoefrt & \bern & \bernrt \\ \hline \hline
\fqrefp{ord}{\thesissub{\$airport='ORD',$\veps=.5$}{'ORD',$0.5$}}  & $21.40$ & $61.58\times$~~(0.35) & $60.17\times$~~(0.36) & $1721.06\times$~~(0.01) & $\mathbf{3093.02\times}$~~(0.01) \\ \hline
\fqrefp{airlinepos}{\$thresh=0} & $46.10$ & $267.75\times$~~(0.17) & $374.92\times$~~(0.12) & $2440.25\times$~~(0.02) & $\mathbf{5135.43\times}$~~(0.01) \\ \hline
\fqrefp{minlatedelay}{\thesisdel{\$min\_dep\_time=}10:50pm} & $28.14$ & $1.19\times$~~(23.58) & $1.74\times$~~(16.14) & $9.57\times$~~(2.94) & $\mathbf{18.58\times}$~~(1.51) \\ \hline
\fqref{ordten} & $21.03$ & $13.38\times$~~(1.57) & $13.64\times$~~(1.54) & $\mathbf{991.50\times}$~~(0.02) & $956.72\times$~~(0.02) \\ \hline
\fqref{negdelay} & $49.15$ & $0.48\times$~~(102.41) & $0.90\times$~~(54.62) & $1.86\times$~~(26.47) & $\mathbf{3.77\times}$~~(13.05) \\ \hline
\fqref{multigroup} & $65.74$ & $1.19\times$~~(55.09) & $1.26\times$~~(52.01) & $12.48\times$~~(5.27) & $\mathbf{21.63\times}$~~(3.04) \\ \hline
\fqref{hpweekbar} & $29.62$ & $0.99\times$~~(29.93) & $1.00\times$~~(29.77) & $2.21\times$~~(13.39) & $\mathbf{2.51\times}$~~(11.79) \\ \hline
\fqref{topkdelay} & $49.31$ & $1.08\times$~~(45.47) & $1.08\times$~~(45.87) & $5.60\times$~~(8.81) & $\mathbf{5.83\times}$~~(8.45) \\ \hline
\fqref{airlinemax} & $46.69$ & $1.16\times$~~(40.33) & $1.34\times$~~(34.81) & $143.84\times$~~(0.32) & $\mathbf{157.94\times}$~~(0.30) \\ \hline
\end{tabular}}
\captionof{table}{\papertext{\scriptsize}Summary of average query speedups and latencies for various error bounders.}
\label{tab:bounder-latencies:dd}
\vspace{-1em}
\end{densecenter}
\thesissub{\end{table*}}{\end{figure}}
}
\papertext{
\begin{table}[t]
\begin{densecenter}
{\scriptsize \begin{tabular}{|l|c||c|c|c|c|} \hline
\multicolumn{1}{|c|}{\cellcolor[HTML]{C0C0C0} \textbf{Query}}
&  \multicolumn{5}{^c|}{\cellcolor[HTML]{C0C0C0}\textbf{Avg Time to Completion (s)}} \\ \hline \hline
                 & \exact & \hoef & \textsf{H}\plusrt & \bern & \textsf{B}\plusrt \\ \hline \hline
\fqrefp{ord}{ORD,$\veps=0.5$}  & $21.40$ & $0.35$ & $0.36$ & $0.01$ & $\mathbf{0.01}$ \\ \hline
\fqrefp{airlinepos}{\$thresh=0} & $46.10$ & $0.17$ & $0.12$ & $0.02$ & $\mathbf{0.01}$ \\ \hline
\fqrefp{minlatedelay}{10:50pm} & $28.14$ & $23.58$ & $16.14$ & $2.94$ & $\mathbf{1.51}$ \\ \hline
\fqref{ordten} & $21.03$ & $1.57$ & $1.54$ & $0.02$ & $\mathbf{0.02}$ \\ \hline
\fqref{negdelay} & $49.15$ & $102.41$ & $54.62$ & $26.47$ & $\mathbf{13.05}$ \\ \hline
\fqref{multigroup} & $65.74$ & $55.09$ & $52.01$ & $5.27$ & $\mathbf{3.04}$ \\ \hline
\fqref{hpweekbar} & $29.62$ & $29.93$ & $29.77$ & $13.39$ & $\mathbf{11.79}$ \\ \hline
\fqref{topkdelay} & $49.31$ & $45.47$ & $45.87$ & $8.81$ & $\mathbf{8.45}$ \\ \hline
\fqref{airlinemax} & $46.69$ & $40.33$ & $34.81$ & $0.32$ & $\mathbf{0.30}$ \\ \hline
\end{tabular}}
\captionof{table}{Summary of query latencies for various error bounders.}
\label{tab:bounder-latencies:dd}
\vspace{-1em}
\end{densecenter}
\end{table}
}

\thesissub{\begin{table}[t]}{\begin{figure}[t]}
\begin{densecenter}
{\papertext{\scriptsize}\begin{tabular}{|l|c||c|c|} \hline
\multicolumn{1}{|c|}{\cellcolor[HTML]{C0C0C0} \textbf{Query}}
&  \cellcolor[HTML]{C0C0C0} & \multicolumn{2}{c|}{\cellcolor[HTML]{C0C0C0}\textbf{Avg Speedup over \scan{} (time in (s))}} \\ \hline \hline
                 & \scan{} (s) & \sync & \peek \\ \hline \hline
\fqrefp{minlatedelay}{\thesisdel{\techsub{10:50pm}{\footnotesize10:50pm}}} & $2.04$ & $1.15\times$~~(1.76) & $\mathbf{1.20\times}$~~(1.69) \\ \hline
\fqref{negdelay} & $45.18$ & $1.11\times$~~(40.74) & $\mathbf{3.43\times}$~~(13.18) \\ \hline
\fqref{multigroup} & $4.10$ & $1.24\times$~~(3.32) & $\mathbf{1.36\times}$~~(3.03) \\ \hline
\fqref{hpweekbar} & $11.05$ & $\mathbf{1.14\times}$~~(9.68) & $1.13\times$~~(9.80) \\ \hline
\fqref{topkdelay} & $47.12$ & $1.40\times$~~(33.76) & $\mathbf{5.35\times}$~~(8.81) \\ \hline
\end{tabular}}
\captionof{table}{Summary of query speedups and latencies for various sampling strategies,
restricted to \groupby queries that take more than $500$ms for \scan with \bernrt.}
\label{tab:sampstrat-latencies:dd}
\end{densecenter}
\thesissub{\end{table}}{\end{figure}}

\begin{figure}[t]
\begin{densecenter}
\resizebox{\linewidth}{!}{\input{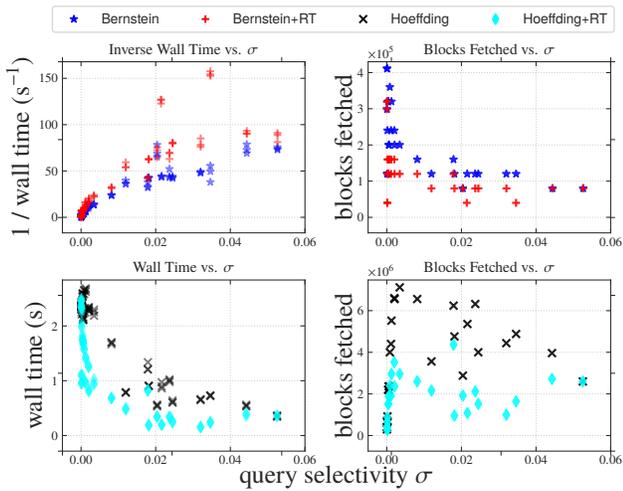}}
\caption{Effect of query selectivity on wall time and blocks fetched,
for selectivity determined by varying the origin airport used to
filter \fqrefp{ord}{$\veps=.5$}. Note that \bern shows inverse
time for clarity.}
\label{fig:vary-sel:dd}
\end{densecenter}
\end{figure}

\begin{figure}[t]
\begin{densecenter}
\resizebox{\linewidth}{!}{\input{\figs/vary-err-and-thresh.pgf}}
\begin{minipage}[t]{.5\linewidth}
\centering
\vspace{-\parskip}\vspace{-\abovedisplayskip}\papertext{\vspace{-1em}}
\subcaption{\label{fig:vary-eps:dd}}
\end{minipage}%
\begin{minipage}[t]{.5\linewidth}
\centering
\vspace{-\parskip}\vspace{-\abovedisplayskip}\papertext{\vspace{-1em}}
\subcaption{\label{fig:vary-having:dd}}
\end{minipage}%
\caption{\textbf{\subref{fig:vary-eps:dd}}:
Effect of requested maximum relative error $\veps$ on
actual relative error achieved for \fqref{ord}.
\textbf{\subref{fig:vary-having:dd}}:
Data required for different \texttt{HAVING} thresholds used in \fqref{airlinepos}.
The group aggregates for also displayed for comparison.}
\label{fig:vary-err-and-thresh:dd}
\end{densecenter}
\end{figure}

\begin{figure}[t!]
\begin{densecenter}
\resizebox{\linewidth}{!}{\input{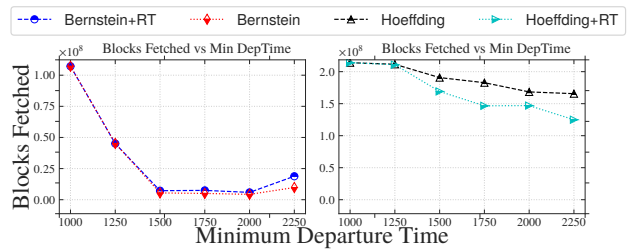}}
\caption{Effect of minimum departure time on blocks fetched
for \fqref{minlatedelay}.}
\label{fig:vary-deptime:dd}
\end{densecenter}
\end{figure}

In this section, we present results of our empirical study.

\subsubsection{Impact of Error Bounder Used} \hfill

\frameme{{\bf \em Summary.} Using the \bernrt error bounder resulted in enormous
speedups (\maxexactspeedup over \exact and \maxhoefspeedup over
\hoef), and additionally was almost always on par or better than
\bern{} ($2\times$ faster under certain conditions).}

We evaluate \hoefrt and \bernrt error bounders, along with \hoef
and \bern{} (to ablate our \rangetrim technique) and an \exact
query processor (to ablate any benefits due to approximation)
against all the queries in \Cref{fig:flightsqueries:dd}, with the
resulting time measurements summarized in \Cref{tab:bounder-latencies:dd}.

We can make a number of interesting observations. First of all, note that
all error bounders incur additional overhead --- in the case
of \fqref{negdelay} where techniques
like \hoef and \hoefrt needed to process all the data in order to terminate
(due to \pma), they actually ran {\em more slowly} than \exact.
Using \bern, which does not suffer from \pma, yielded significant
benefits over \exact, \hoef, and \hoefrt across all queries. In cases
where \hoef and \hoefrt showed improvements over \exact, \bern
amplified these improvements (\fqref{ord}, \fqref{airlinepos}).

Using \rangetrim in conjunction with both \hoef and \bern typically
led to similar performance, with a few queries exhibiting clearly superior
performance (\fqref{negdelay}, \fqref{multigroup}, and \fqref{minlatedelay}).
These queries have the following in common: they all have sparse groups with low selectivity
(either because of the large number of groups in the case of \fqref{negdelay} and \fqref{minlatedelay},
or because of the restrictive filter in the case of \fqref{negdelay}),
and they are all ``easy'' to approximate,
in that none of the groups require too many samples in order to achieve
the relevant stopping condition. (\fqref{topkdelay} also has many groups,
but some of them require many samples due to a large number of airports
with average delay near the max.) This is an ideal condition for \bernrt to
show benefit: sparse groups will bottleneck the query, but \rangetrim
will achieve termination faster since these sparse groups tend to
have fewer outliers than do non-sparse groups. For such bottlenecking sparse
groups, the range bounds
for the \sql{DepDelay} column are overly-conservative
and dominate the sampling complexity. In this case, \bern, which
has \pos, will require twice as many samples for such groups --- and since
these groups are the bottleneck, it will require roughly twice as much time,
an intuition reflected in \Cref{tab:bounder-latencies:dd}.

\subsubsection{Impact of Sampling Strategy Used}

\frameme{{\bf \em Summary.} Using \peek sampling was almost
always better than \sync, in some cases significantly
(more than $3\times$ for \fqref{negdelay} and \fqref{topkdelay}),
while \sync sampling only showed modest gains against simple
\scan-based sampling (up to $1.40\times$, \fqref{topkdelay}).}

\noindent We evaluate the impact of various sampling strategies when used in conjunction
with the \bernrt error bounder, the results of which are summarized
in \Cref{tab:sampstrat-latencies:dd}. In some cases
(\fqref{negdelay} and \fqref{topkdelay}), the performance
of the \scan baseline when used in conjunction with \bernrt
was on par with that of the \exact baseline, indicating that
some form of block skipping can be crucial for queries with \sql{GROUP BY}s.
When implementing active scanning synchronously, however, the improvement
was mediocre across the board, while the active scanning with
lookahead achieved significantly better performance for
\fqref{negdelay} and \fqref{topkdelay}. It is no coincidence that
these are the same queries for which \scan performance using approximation
is similar to \exact performance. This indicates that there were a few
sparse groups preventing termination when \scan is used, which is the
very case for which the greatest benefit can be derived from (an efficient
implementation of) block skipping.

\subsubsection{Impact of Data and Query Characteristics}

\noindent To better understand various data- and query-dependent aspects
of our techniques, we now study the effect of varying the parameters
supplied to \fqref{ord}, \fqref{airlinepos}, and \fqref{minlatedelay}.

\topic{Selectivity $\sigma$ of filter}

\frameme{{\bf \em Summary.} Wall clock time decreases as
the fraction of tuples passing \fqref{ord}'s filter increase,
while blocks fetched first increases, then decreases. \rangetrim
gives the most benefit for filters of intermediate selectivity.}

Different \sql{Origin} attribute values used for filtering
\fqref{ord} have different selectivities. By varying the filter
attribute value, we reveal interesting behavior impacted by the
selectivity of the filter. (We consider selectivity as a number
and not a quality, so that larger proportions of tuples satisfy
predicates with higher selectivity.) For all four error bounding
techniques considered, wall time and blocks fetched are plotted
versus query selectivity in \Cref{fig:vary-sel:dd}.
\bern and \bernrt are plotted separately from \hoef and \hoefrt
for presentation.

As selectivity increases, wall-clock time decreases, as one might
expect, with the benefits of \rangetrim being more obvious for
\hoef than for \bern. The performance gap between techniques with and without \rangetrim generally
decreases with increasing selectivity --- perhaps because filters with
higher selectivity tend to have range bounds that are not as conservative
when compared with the a priori range bounds known to hold for the entire column.

Interestingly, as selectivity increases, the number of blocks fetched first
increases, then decreases. This is likely because the sparsest filters require
examining all the data before terminating, obviating early stopping benefits.
After a certain point, however, early termination kicks in, happening more
quickly as fewer tuples are filtered.

\topic{$\veps$ for stopping condition~\ref{stop:relacc}}

\frameme{{\bf \em Summary.} For different upper bounds on relative
error, the actual relative error in the query result is always within
the requested error, for all error bounders applied to \fqref{ord}.
The achieved relative error drops to 0 more quickly for the more
conservative bounders \hoef and \hoefrt
as the requested error is decreased.}

By varying the requested maximum relative error $\veps$ for
query \fqref{ord}, we reveal its impact on the relative
error achieved for various error bounders, shown in \Cref{fig:vary-eps:dd}.
The main takeaways are that, for all error bounders, the achieved relative
error generally decreases as the requested error bound decreases, with
\hoef-based bounders dropping more quickly, as they are more conservative
due to \pma.

\topic{HAVING threshold for stopping condition~\ref{stop:thresh}}

\frameme{{\bf \em Summary.} \sql{HAVING} thresholds that are closer to
group aggregates require more samples in order to achieve stopping
condition~\ref{stop:thresh}, and \hoef-based error bounders in particular
are more sensitive than \bern-based error bounders for the same
threshold.}

By varying the \sql{HAVING} threshold used to filter groups / airlines
post-aggregate in \fqref{airlinepos} and measuring its effect on the
number of blocks fetched for a particular query, we reveal interesting data-dependent
behavior impacted by the true aggregates for each airline, depicted in
\Cref{fig:vary-having:dd}. This figure also plots the group aggregates
using a horizontal bar chart
sharing the same x-axis as the \sql{HAVING} threshold, revealing that
it is ``harder'' to determine which side of the \sql{HAVING} threshold
a given group is if its aggregate is close to the threshold. Indeed,
from \Cref{fig:vary-having:dd}, we see that the initial thresholds near
$0$ are very easy for all groups, allowing for very fast termination.
The first spike in number of blocks fetched occurs between $6$ and $7$,
corresponding to the aggregate for airline \textsf{NW}. At and after this
point, we see spikes in blocks fetched for both \hoef-based and \bern-based
error bounders whenever the threshold approaches one or more airline aggregate
values, although we note that \bern-based error bounders appear to be more
robust, requiring the threshold to be much closer before they are adversely
affected as compared to \hoef-based bounders.

\topic{Minimum Departure Time for \fqref{minlatedelay}}

\frameme{{\bf \em Summary.} As the minimum departure time is increased,
the spread of average delay between airlines increases, making it easier
to separate the two airlines with the minimum average delays and
achieve stopping condition~\ref{stop:sep} earlier. At the same time,
termination becomes bottlenecked on sparse airlines, increasing the
gap between similar bounders with and without \rangetrim.}

By varying the minimum departure time \mindeptime
in \fqref{minlatedelay}, we reveal its impact on the number of blocks
fetched for various error bounders, shown in \Cref{fig:vary-deptime:dd}.
This plot exhibits two interesting data-dependent behaviors worth unpacking.
First, as the \mindeptime increases, the variance in average
delay between different airlines increases, perhaps because some airlines
tend to have flights that are delayed more for later flights as compared
with other airlines. This makes it easier to achieve stopping
condition~\ref{stop:sep}, since the average delays become more spread out
with increasing minimum departure time, so we observe a decreasing trend
in the number of blocks fetched.
At the same time, as \mindeptime increases, the
selectivity of the various groups decreases. Since all
the groups are sparser, the groups for which stopping condition~\ref{stop:sep}
is bottlenecked are also sparser. Since we have an ``easy'' query (due to
the higher variance between groups) for which sparse groups are bottlenecking
termination, we tend to see a bigger performance gap between bounders with and without
our \rangetrim technique.

\section{Related Work}
\label{sec:related:dd}
In this section, we survey related literature and highlight
similarities and differences with this work.

\topic{Approximate Query Processing (AQP)}
We survey the AQP literature along two dimensions: first, online versus offline;
second, approaches with strong versus asymptotic guarantees.

\emtitle{Online versus Offline AQP.}
Online sampling-based AQP schemes select samples as queries are issued,
contrasted with offline schemes which compute strata ahead of time.
Although our approach does perform a shuffle offline, it is nevertheless closer to
online schemes, as it uses the scramble to compute samples on the fly as
in~\cite{macke2018adaptive,bismarck,qin2014pf,wu2010continuous,zeng2015g}.
Online schemes can use index structures like bitmaps to materialize
relevant samples
on-the-fly~\cite{Hellerstein1997,Kim2015,rahman2016ve,macke2018adaptive},
or obey an accuracy constraint for computing predefined aggregates
without indices~\cite{hou1988statistical,hou1989processing}.
Offline schemes, on the other hand,
materialize samples ahead-of-time~\cite{SPS,blinkdb,agarwal2012blink,icicles}
based off workload assumptions, sometimes tuning the computed strata as new
workload information is available~\cite{blinkdb,icicles}.

While we implement our error bounders without \pma or \pos in the context
of a system for online AQP, our core algorithmic techniques are orthogonal
to the exact approach, and could be paired with either online or offline schemes.

\emtitle{Sample-size-independent versus Asymptotic Guarantees.}
Most of the AQP systems from prior work have traditionally leveraged asymptotic
error bounders~\cite{blinkdb,agarwal2012blink,acharya1999aqua,park2018verdictdb},
though some have mentioned allowing either approach as an option~\cite{Hellerstein1997}.
Other approaches have leveraged
deterministic~\cite{potti2015daq,haas1997large} or concentration-based
error bounding techniques~\cite{SPS,alabi2016pfunk,macke2018adaptive,rahman2016ve,Kim2015}
under range-based or other very mild assumptions.
In some cases, novel asymptotic error bounding techniques have been
developed~\cite{park2018verdictdb,zeng2014analytical,haas1997large} to be used in conjunction with
existing systems. Our approach is analogous to these, but instead of basing
our techniques on asymptotic methods, we develop error bounding techniques
with guarantees independent of sample size, starting from existing concentration-based
methods and systematically ameliorating various pathologies.

\topic{Access Patterns for Informative Samples}
A number of techniques have been developed to optimize access to relevant data
for analytical queries.
\iftech
Sampling-based
approaches~\cite{Kim2015,olken1993random,chaudhuri_optimized_2007,chaudhuri2001overcoming,li_wander_2016,wu2010continuous,macke2018adaptive,Babcock2003,haas1999ripple,icicles,SPS}
attempt to retrieve tuples that will shrink approximation error as quickly as possible.
Index structures such as bitmaps~\cite{Kim2015,macke2018adaptive}
or inverted indexes~\cite{SPS} have been employed to facilitate this, or to simply
accelerate exact analytical queries by quickly retrieving relevant
tuples~\cite{wu2001compressed,chan1998bitmap,needletailoptimal}.
We make use of the sampling engine developed in~\cite{macke2018adaptive}, which leverages
bitmaps and active scanning to adaptively prioritize different groups in the data
while a query is running. While our \rangetrim technique is technically orthogonal
to whatever data access method is employed, it demonstrates the most gains over
existing error bounding techniques when few samples are needed to terminate.
Active scanning is particularly useful for skipping
to data needed to terminate when they are sparse.

Another access strategy worth mentioning explicitly comes from
from~\cite{chaudhuri2001overcoming} and leverages an outlier index.
\else
Please see the technical report~\cite{techreport-dd} for a more extensive survey;
here we focus on two techniques in particular:
{\em outlier indexing} and {\em priority sampling}.

\fi
Outlier indexing~\cite{chaudhuri2001overcoming} works by computing
approximate aggregates derived by combining an estimate from the main table
and an exact aggregate from the so-called ``outlier index'', which stores
all the rows with outlier values.
The benefit of the outlier index is that it shrinks the range of the data
from which samples are taken, allowing for faster convergence of approximate answers.
One could think of the outlier index as an offline analogy of our own \rangetrim technique.
Outlier indexing has some additional limitations that \rangetrim does not have;
namely, it cannot be used to facilitate queries with aggregates involving arbitrary
expressions, since such expressions can drastically change the set of outlying values.
That said, for simple aggregates the two approaches are orthogonal,
and could be leveraged together.

Priority sampling~\techreport{\cite{duffield2007priority,alon2005estimating,thorup2006confidence}}\papertext{\cite{duffield2007priority,alon2005estimating}}
is also particularly useful for coping with outliers. \techreport{If the attribute being aggregated has values $\{w_i\}$,
priority sampling computes $\{\alpha_i\}$ (where $\alpha_i \iidsim \Unif(0, 1)$) and estimates
$\sum_i w_i$ using the subset of the $\{w_i\}$ with the $k$ largest {\em priorities}, where
the priority for the $i$th tuple is given by $w_i / \alpha_i$.}
While priority sampling applies in the presence of arbitrary filters and can furthermore
be modified to allow for computation of \AVG aggregates
in addition to \SUM, it has the drawback that the attribute or expression being aggregated must be known
ahead of time (to say nothing of arbitrary expressions),
so that the tuples can be sorted in descending order of priority,
a limitation our techniques do not have.

\topic{Statistical Estimators and Confidence Intervals}
The well-known error bounders in statistics and probability leverage
asymptotic techniques~\cite{student1908probable,efron1992bootstrap,efron1979bootstrap,hajek1960limiting}\techreport{,}\papertext{.}
\iftech
while those that give strong guarantees independent of sample size
beyond Hoeffding's and Serfling's seminal work~\cite{hoeffding1963probability,serfling1974probability}
are relatively more obscure~\cite{fishman1991confidence,anderson1969confidence}.
We surveyed these in \Cref{sec:problem:dd} when we discussed
the empirical Bernstein-Serfling error bounder developed by
Bardenet et al.~\cite{bardenet2015concentration}, which we adapt for use
in a database setting with our \rangetrim technique.
\else
We already surveyed relatively more obscure \ssi bounders
in \Cref{sec:problem:dd} when we discussed
the empirical Bernstein-Serfling error bounder developed by
Bardenet et al.~\cite{bardenet2015concentration}, which we adapt for use
in a database setting with our \rangetrim technique.
\fi

\section{Conclusion and Future Work}
\label{sec:conclusion:dd}

We categorized existing conservative error bounders in terms of
two pathologies, \pma and \pos, and developed a technique, \rangetrim,
for eliminating \pos from any range-based error bounder.
\iftech
We showed the advantage of using the empirical Bernstein-Serfling bounder in the
context of a real system we are developing, \fv, that accelerates
approximate queries significantly over a Hoeffding-Serfling-based error bounder,
which suffers from \pma. We furthermore showed that augmenting this error
bounder with our \rangetrim technique leads to an additional $2\times$ in
the best case, without ever hurting performance in the worst case.
\else
We showed how an \ssi Bernstein-based bounder without \pma
can significantly accelerate approximate queries,
and how our \rangetrim technique, which eliminates \pos, leads to an additional
$2\times$ speedup in the best case, without ever hurting performance in the worst case.
\fi
By implementing our distribution-aware techniques in the context of \fv,
which
\iftech
is aware of practical considerations such as locality, optional stopping, and
block skipping in order to prioritize
\else
prioritizes
\fi
groups that require more samples in order to facilitate
early termination, we demonstrate significant speedups (\maxexactspeedup over exact processing
and \maxhoefspeedup over traditional techniques based on Hoeffding) without losing guarantees.
This suggests a viable path toward practical with-guarantees \thesissub{AQP for workload-agnostic
analytics\papertext{.}\techreport{; future work could include, for example, the development
of an optimizer that intelligently determines when to leverage traditional
data layouts and index structures for exact query processing and when to
leverage a scramble for approximate results with exact quality.}}{data analytics.}

\ifvldb
\bibliographystyle{abbrv}
\else
\bibliographystyle{IEEEtran}
\fi
{\papertext{\scriptsize}
\iftech
  \ifvldb
  \bibliography{main}
  \else
  \bibliography{IEEEabrv,main}
  \fi
\else
  \ifvldb
  \bibliography{main}
  \else
  \bibliography{IEEEabrv,main-compact}
  \fi
\fi
}

\techreport{\clearpage \nobalance \appendix
\label{sec:appendix:dd}

\iftech

\section{R\lowercase{ange}T\lowercase{rim} B\lowercase{ounder} P\lowercase{seudocode}}
\label{sec:rtbounder:dd}

\def\inner{\texttt{inner}}
\def\undef{\texttt{undefined}}
\thesissub{\begin{algorithm}[h]}{\begin{algorithm}[t]}
\caption{\rangetrim error bounder}\label{alg:rangetrimbounder:dd}
\algrenewcomment[1]{\texttt{/*} #1 \texttt{*/} \\}
\SetKwInOut{Input}{Input}
\SetKwInOut{Output}{Output}
\SetKwRepeat{Do}{do}{while}
\SetKwBlock{Repeat}{Repeat}{}
\SetKwProg{myfunc}{function}{}{}
\SetKwFunction{init}{init\_state}
\SetKwFunction{update}{update\_state}
\SetKwFunction{lbound}{Lbound}
\SetKwFunction{rbound}{Rbound}
{\thesisdel{\scriptsize}
\Input{Inner \ssi range-based error bounder \inner}
\nonl\ \\
\myfunc{\init{}}{
\Return{\upshape{\texttt{\big\{\\
\quad$S_\ell$: \inner.\init{},\\
\quad$S_r$: \inner.\init{},\\
\quad$a'$: \undef,\\
\quad$b'$: \undef\\
\big\}}}}\;
}
\nonl\ \\
\myfunc{\update{$S, v$}}{
\uIf{$a', b'$ \upshape{are \undef}}{
	$S_\ell' \gets S_\ell$\;
	$S_r' \gets S_r$\;
	$a'' \gets v$\;
	$b'' \gets v$\;
}\uElse{
	$S_\ell' \gets$ \inner.\update{$S_\ell, \min{(v, b')}$}\;
	$S_r' \gets$ \inner.\update{$S_r, \max{(v, a')}$}\;
	$a'' \gets \min{(a', v)}$\;
	$b'' \gets \max{(b', v)}$\;
}
\Return{\upshape{\texttt{\big\{$S_\ell$: $S_\ell'$, $S_r$: $S_r'$, $a'$: $a''$, $b'$: $b''$\big\}}}}\;
}
\nonl\ \\
\myfunc{\lbound{$S, a, b, N, \delta$}}{
\Return{\upshape{\inner.\lbound{$S.S_\ell, a, S.b', N-1, \delta$}}}\;
}
\nonl\ \\
\myfunc{\rbound{$S, a, b, N, \delta$}}{
\Return{\upshape{\inner.\rbound{$S.S_r, S.a', b, N-1, \delta$}}}\;
}
}
\end{algorithm}

We give an implementation of our \rangetrim
technique in terms of the interface from \Cref{sec:state:dd}
in \Cref{alg:rangetrimbounder:dd}.

\section{H\lowercase{andling} A\lowercase{rbitrary} E\lowercase{xpressions}}
\label{sec:exprrange:dd}
In this \work, we assumed that column $c_i$ was known
to lie in some range $[a_i, b_i]$. We then showed how
to feed these bounds into our \rangetrim procedure to compute
conservative CIs for, \eg, \AVG($c_i$). In general, however,
we may want to compute an aggregate involving an arbitrary
expression in terms of several columns. That is, we may want
to compute CIs for, \eg, \AVG($f(c_1,\ldots,c_n)$).
We now show how to do so for a large class of $f$
by optimizing over such $f$ (while using the
individual bounds $[a_i, b_i]$ for each column as constraints)
in order to compute {\em derived} range bounds of the form
\[ \left[\inf_{c_1,\ldots,c_n}f(c_1,\ldots,c_n), \sup_{c_1,\ldots,c_n}f(c_1,\ldots,c_n)\right]\]

\topic{Applicable Expressions}
To compute a derived lower range bound, we need to be able to
either solve or compute a lower bound for the following optimization
problem:
\begin{align*}
\min_{c_1,\ldots,c_n} \quad & f(c_1,\ldots,c_n)\\
\textrm{\sthat} \quad & a_i \leq c_i \leq b_i, \quad \forall 1 \leq i\leq n
\end{align*}
The case for the derived upper range bound is analogous, but with
$f$ replaced by $-f$. We show how to compute both lower and upper
derived bounds under two kinds of conditions: (i) {\em the monotonicity condition}; \ie, $f$ is monotone in each $c_i$,
and (ii) {\em the convexity condition}; \ie, either $f$ or $-f$ is convex.
This handles a large number of expressions in practice.

\emtitle{1. Expressions Monotone in each Column.}
If $f$ is monotone in each column $c_i$, one simply needs to
check whether $a_i$ or $b_i$ gives the smaller (\resp larger)
value when computing the lower (\resp upper) bound,
and evaluate $f$ on the boundaries for each of these cases.

\emtitle{2. Convex or Concave Expressions.}
Without loss of generality, we now consider the case of convex $f$.
A large body of existing work focuses on minimizing a convex function subject to convex constraints;
please see Boyd \etal~\cite{Boyd2004} for relevant background.
In our case, the constraints are all linear (and are sometimes referred to as ``box'' constraints),
and most kinds of convex functions in practice can be optimized efficiently with off-the-shelf software
under such constraints, so we do not go into detail here.

Maximizing a $f$ under box constraints is more difficult. Since $f$ is convex, the maximum (and therefore
the derived upper range bound we seek) will occur at some set of boundary points; \ie, if $a_i\leq c_i\leq b_i$,
we know that the maximum will occur at one of $c_i=a_i$ or $c_i=b_i$. If we have $n$ columns involved
in the expression $f$, however, we will require evaluating $f$ on all $2^n$ combinations of boundary points
for the constraints. Fortunately, database aggregates over expressions typically do not involve more than $2$
or $3$ columns, and any $n\leq 20$ or so can be handled without trouble.

\begin{example}
Suppose the user issues a query to compute \AVG($(2c_1 + 3c_2 - 1)^2$) involving columns $c_1$ and $c_2$,
where we have range bounds $c_1\in[-3, 1]$ and $c_2\in[-1, 3]$. The minimum of $(2c_1 + 3c_2 - 1)^2$ subject
to these constraints is simply $0$, and can be found via quadratic programming.
The maximum can be obtained
by checking the boundaries $(-3, -1)$, $(-3, 3)$, $(1, -1)$, and $(1, 3)$, and we see that it occurs
at $(1, 3)$, for which $(2\cdot1 + 3\cdot3 - 1)^2 = 100$; thus, the derived range bounds will be $[0, 100]$.
\end{example}

\section{P\lowercase{roof of} T\lowercase{heorem}~\ref{THM:DKWNOREPLACE:DD}}
\label{sec:dkwnoreplace:dd}
In \Cref{sec:bounders:dd}, we claimed that the DKW inequality holds for sampling
without replacement from a finite population; we now sketch the proof.

\dkwnoreplace*
\begin{proof}
Sketch: following the original paper from Wald and Wolfowitz
on confidence limits for CDFs~\cite{wald1939confidence},
it suffices to consider the CDF for mass distributed uniformly
at each integer $1, 2, \ldots, N$. For each without-replacement
sample size $m$ and each deviation $\veps$, we would like to be
able to claim that
\[ \Parg{\sup|F_N - \ecdf_{N,m}| \geq \veps} < \Parg{\sup|F_{N'} - \ecdf_{N',m}| \geq \veps} \]
for every $N' > N$ --- that is, in some sense, the CDF becomes monotonically ``harder''
to estimate as we increase the dataset size. Unfortunately, this turns out to
not be the case, but in fact the claim follows if we merely prove the weaker condition
that
\[ \Parg{\sup|F_N - \ecdf_{N,m}| \geq \veps} < \Parg{\sup|F_{N'} - \ecdf_{N',m}| \geq \veps} \]
for {\em infinitely many} $N' > N$, implying that, as $N' \rightarrow \infty$,
the resulting probability to which $\Parg{\sup|F_{N'} - \ecdf_{N',m}|}$ converges
(necessarily bounded by the probability computed in the DKW inequality)
is an upper bound for the corresponding probability at every
finite $N'$, from which the claim would follow.

We show this via construction: namely, we show that, for every $N$, $m$, and $\veps$,
\[ \Parg{\sup|F_N - \ecdf_{N,m}| \geq \veps} < \Parg{\sup|F_{2N} - \ecdf_{2N,m}| \geq \veps} \]
That is, the CDF becomes monotonically harder to estimate each time we double the dataset size.
To show this, we consider two cases. Case 1: if point $2i-1$ is sampled, then point $2i$ is not
sampled, and vice versa for every $i=1,\ldots,N$. Conditioned on this event, the probability that
$\sup|F_{2N} - \ecdf_{2N,m}| \geq \veps$ at least the (unconditioned) probability that
$\sup|F_N - \ecdf_{N,m}| \geq \veps$, since samples at odd indices can only increase the deviation,
and samples at even indices cannot decrease it. Case 2: there is at least one $i$ for which points
at both indices $2i-1$ and $2i$ are sampled. Each such point is conceptually similar to
reducing $m$ by 1 in the original dataset of size $N$, but randomly weighting one of the samples
by 2 instead of 1. It can be shown that each time this is done, the probability that
$\sup|F_N - \ecdf_{N,m}| \geq \veps$ increases.
\end{proof}

\fi
}

\end{document}